\documentclass{article}

\usepackage[english]{babel}
\usepackage{amsthm,amsmath,lipsum}
\usepackage[algo2e]{algorithm2e} 
\usepackage{xcolor}
\usepackage{fullpage}
\usepackage{amssymb}
\usepackage{centernot}
\usepackage{tikz}
\usepackage{soul}
\usetikzlibrary{shapes.geometric, arrows}
\usetikzlibrary{positioning}
\usetikzlibrary{calc}
\usepackage{dirtytalk}

\DeclareFontFamily{U}{mathb}{\hyphenchar\font45}
\DeclareFontShape{U}{mathb}{m}{n}{
      <5> <6> <7> <8> <9> <10> gen * mathb
      <10.95> mathb10 <12> <14.4> <17.28> <20.74> <24.88> mathb12
}{}
\DeclareSymbolFont{mathb}{U}{mathb}{m}{n}
\DeclareMathSymbol{\llcurly}{3}{mathb}{"CE}
\DeclareMathSymbol{\ggcurly}{3}{mathb}{"CF}

\usepackage{thm-restate}
\theoremstyle{plain}

\newtheorem{theorem}{Theorem}
\newtheorem{lemma}{Lemma}
\newtheorem{definition}{Definition}

\newtheorem{corollary}{Corollary}
\newtheorem{claim}{Claim}

\usepackage{times}
\usepackage{soul}
\usepackage{url}
\usepackage[colorlinks=true, allcolors=blue]{hyperref}
\usepackage[utf8]{inputenc}
\usepackage[small]{caption}
\usepackage{graphicx}
\usepackage{amsmath}
\usepackage{amsthm}
\usepackage{amssymb}
\usepackage{booktabs}
\usepackage{algorithm}
\usepackage{algorithmic}
\usepackage{color}
\usepackage[switch]{lineno}

\title{Improving the Price of Anarchy via Predictions in Parallel-Link Networks}

\author{George Christodoulou\thanks{Aristotle University of
    Thessaloniki \& Archimedes, Athena Research Center, Greece. Email: \texttt{\{gichristo,vgchrist\}@csd.auth.gr} }
\and{Vasilis Christoforidis\footnotemark[1]}
\and Alkmini Sgouritsa\thanks{Athens University of Economics and Business \&  Archimedes, Athena Research Center, Greece. Email: \texttt{\{alkmini,ioa.vlahos\}@aueb.gr}}
\and {Ioannis Vlachos\footnotemark[2]}}

\date{}

\begin{document}
\maketitle

\begin{abstract}

We study non-atomic congestion games on parallel-link networks with affine cost functions. We investigate the power of machine-learned predictions in the design of {\em coordination mechanisms} aimed at minimizing the impact of selfishness. Our main results demonstrate that enhancing coordination mechanisms with a simple advice on the input rate can optimize  the social cost whenever the advice is accurate (\emph{consistency}), while only incurring minimal losses even when the predictions are arbitrarily inaccurate (\emph{bounded robustness}). Moreover, we provide a full characterization of the consistent mechanisms that holds for all monotone cost functions, and show that our suggested mechanism is optimal with respect to the robustness. We further explore the notion of smoothness within this context: we extend our mechanism to achieve error-tolerance, i.e. we provide an approximation guarantee that degrades smoothly as a function of the prediction error, up to a predetermined threshold, while achieving a bounded robustness. 
\end{abstract}

\section{Introduction}

A fundamental issue in large decentralized systems is the impact of selfish behavior, wherein each player selects their strategy to minimize their own cost without considering the overall system efficiency. The decentralized decision-making can lead to suboptimal outcomes for the population as a whole. In this context, the celebrated notion of the Price of Anarchy (PoA) \cite{KoutsoupiasP99} is the standard metric that quantifies the inefficiency arising due to players' selfish behavior. It captures the ratio of the social cost in the worst Nash equilibrium to the optimal social cost. In this work, we consider congestion games, which is a fundamental resource allocation problem, from the perspective of anticipating players' strategic behavior in order to improve overall system performance. 

Parallel-link networks serve as a fundamental model for studying congestion phenomena, where each link represents an independent resource and players choose among them. Analyzing parallel-link networks yields key insights into the performance of mechanisms in more complex networked environments. 
Moreover, this model captures a wide range of very well-studied games, including load balancing and scheduling games, where tasks must be assigned to machines with potentially different delays. More precisely, there is a set of available machines and a set of users, each associated with an infinitesimally small task that needs to be processed by a machine. Each machine is characterized by a cost function that depends on the congestion that appears on that machine. Each user selects a strategy, i.e., choosing a machine, by seeking to minimize their individual completion time. 

To mitigate the impact of selfish behavior, we explore the use of coordination mechanisms (CMs) \cite{tcs/ChristodoulouKN09:CMs, algorithmica/ChristodoulouMP14} that can be seen as an intervention to improve the overall system performance. Additionally, motivated by the observation that additional information on the input can have a substantial impact on the quality of induced equilibria, we study the effect of exogenous information provided in the form of machine-learned predictions on the design of CMs aimed at steering the players towards more efficient outcomes. This approach of enhancing algorithmic decision-making with exogenous information has had a profound impact in recent algorithmic research \cite{jacm/LykourisV21,cacm/MitzenmacherV22}. In line with this direction, a closely related work by \cite{sigecom/GkatzelisKOlliasSgouritsaTan22} incorporates predictions about the anticipated demand in the design of learning-augmented cost-sharing mechanisms in scheduling games.

\subsection{Our Results}

We study learning-augmented CMs in the context of network non-atomic congestion games over parallel-link networks with affine cost functions. Given a prediction of the total demand we design {\em consistent} CMs that achieve the optimal performance with respect to the PoA when the prediction is accurate, while remaining near-optimal even under arbitrarily erroneous predictions, i.e., they have bounded robustness. More precisely:

\begin{itemize}
    \item In Section~\ref{sec:CMs} we provide a thorough study of consistent CMs. First, we provide a full characterization of consistent CMs for any monotone cost function (Lemma~\ref{lemma:characterizationCons}). We then propose a mechanism that applies the minimal necessary modification which we call \textsc{MinCharge} and show that it is consistent and also $2$-robust (Thm~\ref{thm:Optimality}). We then utilize our characterization to show that any consistent mechanism must be at least $2$-robust (Lemma~\ref{lemma:LBAffine}), establishing the optimality of \textsc{MinCharge} among all consistent CMs. 
    \item In Section~\ref{sec:smoothness} we study the tradeoff between robustness and {\em smoothness} in the design of consistent CMs. Smoothness is a desirable property of learning-augmented algorithms, which ensures that the performance of the mechanisms degrades gracefully as a function of the prediction error. First, we extend the characterization of consistent CMs to capture \emph{error-tolerance}, ensuring that the approximation guarantee degrades smoothly as a function of the prediction error, up to some specified error-tolerance threshold (Lemma~\ref{lemma:characterizationConsSmooth}). We then propose the \textsc{ErrorTolerant} mechanism and analyze its approximation and robustness for the case of multiple links (Thm~\ref{thm:errorTolerance}).  For the case of two links, we also provide the tight approximation guarantee when the prediction error does not exceed the error-tolerance threshold.
\end{itemize}

\subsection{Related Work} 

 We consider the problem of routing traffic to improve the performance of a network under congestion. We assume that in absence of a central authority, users behave selfishly to minimize their cost. As generally equilibria do not minimize the social cost, we adopt the celebrated notion of the Price of Anarchy to analyze the performance degradation due to selfish behavior \cite{KoutsoupiasP99}. A vast amount of literature has  followed this notion and obtained tight bounds for a wide variety of games, e.g.,  \cite{RoughgardenT02,Roughgarden03Indep}. A lot of attention has been given specifically to networks of parallel links \cite{mor/AcemogluOzdaglar07, eor/HarksSV19}. \cite{GairingLuckingMavronicolasMonien10} studied the problem of routing $n$ users on $m$ parallel machines. \cite{DBLP:journals/mor/FalkenhausenH13,GkatzelisPountourakisSgouritsa21,ChristodoulouGkatzelisSgouritsa} considered resource-aware cost-sharing protocols for scheduling games.

{\bf Interventions.} Coping with the effects of selfish behavior has been a central theme in algorithmic game theory.  
A prominent line of research explored the use of taxation mechanisms and tolls to incentivize agents toward more efficient outcomes --either by incorporating them in the total disutility experienced by each player, or treating them as internal transfers within the system that can be refunded later. It is well-known that marginal taxes induce the optimal flow for the original cost functions as a Nash flow in the modified network (e.g. see \cite{stoc/ColeDR03,focs/FleischerJM04,focs/KarakostasK04}). In this work, we focus on the concept of CMs;\footnote{CMs can be seen as flow-dependent local taxes on each resource, where the cost functions are allowed to increase in an arbitrary manner.} originally introduced in the context of scheduling games by \cite{tcs/ChristodoulouKN09:CMs}, CMs are implemented as local policies on each machine that dictate the order of processing tasks, and has since been extended to a variety of settings \cite{tcs/ImmorlicaLMS09,ColeCGMO15,Kollias13,algorithmica/ChristodoulouMP14,BhattacharyaIKM14,AzarFJMS15}. 

We underline that a fundamental element in our analysis is the fact that CMs work for any given demand, i.e., they do not assume foreknowledge of the total demand in the network as is required for marginal edge pricing. \cite{icalp/Colini-Baldeschi18} examined whether optimal flows can be induced via demand-independent tolls for any travel demand, albeit without performance guarantees given on the PoA. 
In contrast, we leverage additional information in the form of a demand prediction to induce optimal flows at minimal cost when the prediction is accurate. \cite{KarakostasKolliopoulos04} showed that for affine cost functions and fixed demand rate, enforcing the optimal flow with constant edge taxes results in a PoA of $2$ when both the tax and the latency contribute to the social cost. Perhaps most closely related to our work is the work of \cite{algorithmica/ChristodoulouMP14} which studies coordination mechanisms for parallel-link networks. They propose a mechanism in which the cost of the resulting Nash flow over the original total latency is strictly less than $4/3$, while they also provide a tight result for the case of two links, demonstrating the existence of a coordination mechanism that is $1.19$-competitive. \cite{isaac/WangYC16:ImprovedTax} obtained comparable improvements using alternative tax schemes. 

In this work, we take a further step by augmenting the mechanism with predictions about the total demand expected to flow through the network. Our work is at the intersection of two main threads: the design of interventions in congestion games, and the use of exogenous information in the form of machine-learned predictions--commonly referred to as the \emph{learning-augmented framework}. 

{\bf Algorithms with Predictions.} In recent years, machine-learned predictions have been widely adopted in designing algorithms \cite{sigecom/MahdianNS07,jacm/LykourisV21,cacm/MitzenmacherV22}. This approach aims to circumvent the limitations of worst-case analysis that has been dominant in algorithm analysis. More closely related to our work, is the line of work on learning-augmented mechanisms in game-theoretic settings, specifically on strategic scheduling \cite{ijcai/XuLu22,ITCS/BalkanskiGT23}, network formation games \cite{sigecom/GkatzelisKOlliasSgouritsaTan22}, facility location games \cite{ijcai/XuLu22,mor/AgrawalBGOT24,nips/BalkanskiGS24:randomizedFL,nips/Barak0T24:MAC}, mechanism design \cite{ijcai/XuLu22,sigecom/Colini-Baldeschi24,neurips:CSV24}, and auctions \cite{sigecom/LuW024:auctions,sigecom/BalkanskiGTZ24:OnlineMD,ijcai/CaragiannisK24,soda/Gkatzelis0T25}.

\section{Preliminaries}
\label{sec:prelims}

{\bf The setting.} We focus on the well-studied selfish routing problem, which serves as a clear and illustrative example of non-atomic congestion games. We consider single-commodity congestion games on parallel-link networks, defined by a directed graph with only two vertices $s,t$ representing the source and sink, respectively, and a set $M$ of $m$ links that have direction from $s$ to $t$. Each link $i\in M$ comes with a non-decreasing, non-negative latency functions $\ell_i(\cdot)$, that describes the travel time as a function of the amount of traffic on that particular link. We denote by $r$ the units of traffic/demand routed from $s$ to $t$ and we call it {\em input rate}. A flow $f=(f_i)_{i\in M}$ defines for each link $i\in M$ the amount of flow passed through $i$, which is denoted by $f_i$; it holds that $\sum_{i \in M} f_i = r$. Then, $\ell_i(f_i)=\ell(f)$ denotes the latency on link $i\in M$ under flow $f$. Naturally, we define the total latency of flow $f=(f_i)_{i\in M}$ as $C(f) = \sum_{i \in M} f_i \ell_i(f_i)$. For any input rate $r$, we denote as $C_{opt}(r)$ the cost of the optimal flow, i.e., the minimum cost of any flow at rate $r$. 

{\bf Price of Anarchy (PoA).} A feasible flow $f$ that routes $r$ units of flow from $s$ to $t$ is a Nash (or Wardrop) equilibrium, or simply a Nash flow, if and only if for any two links $i, j \in M$ with $f_{i} > 0$, $\ell_{i}(f) \le \ell_{j}(f)$. That is, all used links (with positive flow) have equal cost under a Nash flow, which makes the Nash flows to be essentially unique in terms of the edge costs \cite{Beckmann1956-}. We use $C_N(r)$ to denote the cost of any Nash flow at input rate $r$. Then, the Price of Anarchy (PoA) for input rate $r$ and the PoA for {\em any} rate are defined, respectively, as:

$$PoA(r) = \frac{C_N(r)}{C_{opt}(r)} \quad \text{and} \quad PoA = \max_{r>0}PoA(r)\,.$$

The tight $4/3$ PoA bound is already attained in a simple network consisting of two parallel links with affine latency functions. This network is commonly referred to as Pigou's network and the corresponding PoA bound as Pigou's bound \cite{RoughgardenT02}. 

{\bf Affine latency functions.} An affine latency function is of the form $\ell(x)=ax+b$, for $a,b\geq 0$. We denote by $\ell_i(x)=a_ix+b_i$ (with $a_i,b_i>0$) the latency function of link/edge $i\in \{1,\ldots,m\}$. We remark that if two links $i,j$ have equal constant factors $b_i=b_j$ then replacing them with one link of latency function $\ell(x)=\frac{a_ia_j}{a_i+a_j}x+b_i$ gives an equivalent network in the following sense: a flow $f$ is an equilibrium flow (optimal flow, respectively) in the initial latencies if and only if it is an equilibrium (optimal) flow under the merged latencies by adding the respective link flows (see Lemma~\ref{auxLem0}). Therefore, w.l.o.g. we assume that all constant terms are distinct and we rename the links in increasing order of their constant term, i.e., $0 \le b_1 < b_2 < ... < b_k$. Moreover, we may assume that there is at most one link with $a_i=0$, as in the case of multiple constant-latency links, only the one with the smallest $b_i$ would ever be used in any Nash or optimal flow.

{\bf Coordination Mechanisms (CMs).} A Coordination Mechanism (CM) is defined as a set of modified latencies $\hat{\ell}=(\hat{\ell}_i)_{i \in M}$ such that for any $i\in M$, $\hat{\ell}_i$ is strictly increasing\footnote{Under strictly increasing latencies, the CM admits a unique Nash flow in link loads.} and $\hat{\ell}_i(x) \ge \ell_i(x)$ for all $x \ge 0$, and  \cite{tcs/ChristodoulouKN09:CMs,algorithmica/ChristodoulouMP14}. 
Let $\hat{C}_{N}(r)$ be the maximum cost of any Nash flow at rate $r$ under the modified latency functions, i.e., the cost of the worst case Nash flow when $\hat{\ell}$ is used rather than $\ell$. Then, 
the engineered (modified) PoA, termed ePoA, for input rate $r$, and the ePoA for {\em any} rate are defined, respectively, as

$$ePoA(r) = \frac{\hat{C}_{N}(r)}{C_{opt}(r)} \quad \text{and} \quad ePoA=\max_{r>0} ePoA(r)\,.$$

Note that the optimum refers to the minimum cost routing with respect to the \emph{original} latency functions, whereas the cost of the Nash flow includes any increment applied on the latency functions by the CM. 

{\bf Learning-augmented framework.} We consider a prediction on the input rate, $\bar{r}$, that can be seen as an estimate of the expected traffic that is going to flow in the network. We want to design CMs that leverage the prediction to obtain improved ePoA bounds. The induced cost of a Nash flow and the ePoA of a CM now depends on the predicted input rate $\bar{r}$, therefore we write:

$$ePoA(r,\bar{r}) = \frac{\hat{C}_{N}(r,\bar{r})}{C_{opt}(r)} \quad \text{and} \quad ePoA=\max_{r,\bar{r}>0} ePoA(r,\bar{r})\,.$$

When the prediction is correct, we want to induce the optimal routing with minimum cost for this input rate $\bar{r}$, without sacrificing worst-case guarantees, even if the given prediction is arbitrarily wrong. Given a prediction $\bar{r}$, we say that the CM is {\em consistent} if it achieves $ePoA(\bar{r},\bar{r})=1$, for any $\bar{r}>0$, and it is {\em $\rho$-robust} if it achieves $ePoA\leq \rho$.  

{\bf Prediction error - Smoothness.}
    We further use the prediction error to measure the inaccuracy of the prediction. Given a non-atomic congestion game with a predicted input rate $\bar{r}$ and actual input rate $r$, the \emph{prediction error} is given by

    $$\eta(r,\bar{r}) = \lvert r - \bar{r} \rvert.$$

 A desirable property for a mechanism is to be smooth with respect to the prediction error. We say that a mechanism is {\em smooth} if its approximation guarantee degrades gracefully with the prediction error. More precisely, a CM is  smooth if its ePoA$(r,\bar{r})$ is continuous as a function of $r$ for any $\bar{r}>0$. Moreover, we say that a CM is error-tolerant if for any $\bar{r}>0$, ePoA$(r,\bar{r})$ is continuous for any $r\in [\bar{r}-\bar{\eta}, \bar{r}-\bar{\eta}]$, i.e., when $\eta(r,\bar{r})\leq \bar{\eta}$.

{\bf Equilibrium existence.}
We use the term Nash flow to refer to the equilibrium concept used throughout the paper; for the connection between Nash flows and Wardrop equilibria see \cite{HaurieMarcotte1983}. In our setting, we allow latency functions to be discontinuous, since it is known that the use of continuous modified latency functions yields no improvement \cite{algorithmica/ChristodoulouMP14}. However, standard Wardrop equilibria may not always exist under discontinuous latencies. To address this, we adopt the notion of User Equilibria \cite{transci/BernsteinS94:UserEq}, which aligns with Wardrop's first principle, i.e., no individual commuter can unilaterally improve their travel time by switching routes. Formally, we say that a feasible flow $f$ that routes $r$ units of flow is a User Equilibrium, if and only if for any two links $i, j \in M$ with $f_{i}>0$, 

$$\ell_{i}(f) \le \liminf_{\epsilon \downarrow 0}\ell_{j}(f + \epsilon\boldsymbol{1}_{j} -\epsilon\boldsymbol{1}_{i}), $$

where $\boldsymbol{1}_i$ denotes a unit of flow passing through link $i$. Under lower-semicontinuous, non-decreasing latencies, the existence of user equilibria is guaranteed  \cite{transci/BernsteinS94:UserEq}, and we henceforth refer to them as simply Nash flows. 
For a more thorough treatment, we refer the reader to \cite{NesterovDePalma,patriksson2015traffic, MarcottePatriksson}.

\section{A Learning-Augmented Coordination Mechanism for Parallel-Link Networks}
\label{sec:CMs}

In this section, we first derive a characterization of all consistent mechanisms (Lemma \ref{lemma:characterizationCons}). Then, we define a meaningful class of parametrized mechanisms that satisfy the necessary and sufficient conditions of the characterization. We propose a mechanism within this class that applies the minimal required modification, and we show that it is $2$-robust (Theorem \ref{thm:Optimality}). Finally, we establish a matching lower bound, proving that no consistent mechanism can be $(2-\epsilon)$-robust, for any $\epsilon>0$ (Lemma \ref{lemma:LBAffine}).

\subsection{Characterization of Consistent Mechanisms}

We begin with the characterization of all consistent mechanisms. Notably, this result applies to all monotone latency functions, and not merely affine ones.

\begin{lemma} [Characterization lemma]
\label{lemma:characterizationCons}
    Given any set of arbitrary monotone latency functions $(\ell_i)_{i \in M}$, a predicted input rate $\bar{r}$, and any CM, $(\hat{\ell}_i)_{i \in M}$, we denote by $\bar{f}^*=(\bar{f}_i^*)_{i\in M}$ the optimal flow for input rate $\bar{r}$ and we define the maximum latency on any used link under $\bar{f}^*$ as $L = \max_{i: \bar{f}_i^*>0} \ell_i(\bar{f}_i^*)$. Then, the CM is consistent, if and only if for every link $i$ with $\bar{f}_i^*>0$:

    \begin{itemize}
        \item $\hat{\ell}_i(\bar{f}_i^*) = \ell_i(\bar{f}_i^*)$ and
        \item $\liminf_{\epsilon \to 0^+} \hat{\ell}_i(\bar{f}_i^* + \epsilon) \ge L$
    \end{itemize}
\end{lemma}

\begin{proof}

We first show that any CM that satisfies the two conditions of the statement is consistent. For this, suppose that the input rate is $r=\bar{r}$, then we show that for any Nash flow $f=(f_i)_{i\in M}$, $C(f)=C(\bar{f}^*)$. 
If $f$ and $\bar{f}^*$ are the same then trivially $C(f)=C(\bar{f}^*)$. So, suppose that there exists a link $i$ with $0\leq f_i<\bar{f}_i^*$. Therefore, there exists some link $j$ with $f_j>\bar{f}_j^*$. If $\bar{f}_j^*>0$, by the second condition of the lemma's statement, it holds that $\hat{\ell}_j(f_j)\geq L$. If $\bar{f}_j^*=0$, $b_j\geq L$, otherwise the optimal flow would use link $j$ as well; therefore, it is also $\hat{\ell}_j(f_j) \geq b_j \geq L$ in this case. Overall, it holds $\hat{\ell}_j(f_j)\geq L \geq \hat{\ell}_i(\bar{f}_i^*)\geq \hat{\ell}_i(f_i)$, where the second inequality comes from the definition of $L$. It cannot be $\hat{\ell}_j(f_j)> \hat{\ell}_i(f_i)$, since $f_j>0$, and this would violate the fact that $f$ is a Nash flow. Therefore, for any such links $i,j$ it holds that $\hat{\ell}_j(f_j)= \hat{\ell}_i(f_i)$, which results in $C(f)=C(\bar{f}^*)$.

Now we show that for any consistent CM the two lemma's conditions hold. Suppose again that the input rate is $r=\bar{r}$. Since all latency functions $(\ell_i)_{i\in M}$ are strictly increasing apart maybe from one, it is easy to check that the optimal flow $\bar{f}^*$ is unique. Therefore, in order any CM to be consistent it needs to enforce $\bar{f}^*$ as the Nash flow and not to increase the latencies on that flow. This proves the necessity of the first condition in the lemma. Regarding the second condition, suppose for the sake of contradiction that for some link $i$ with $\bar{f}_i^*>0$, $\liminf_{\epsilon \to 0^+} \hat{\ell}_i(\bar{f}_i^* + \epsilon) < L$. We show that $\bar{f}^*$ is not a Nash flow. Let $k=\arg\max_{j: \bar{f}_j^*>0} \ell_j(\bar{f}_j^*)$ be the link with the maximum latency under $\bar{f}^*$. Then, it holds that $\bar{f}_k^*>0$, and $\ell_k(\bar{f}_k^*)=L>\liminf_{\epsilon \to 0^+} \hat{\ell}_i(\bar{f}_i^* + \epsilon)$, which means that $\bar{f}^*$ is not a Nash flow and the CM is not consistent. Thus, the second condition in the lemma is also necessary.
\end{proof}

We further give the following lemma that provides useful properties to be crucial in the next sections. For this we define $k(\bar r)$ for some predicted input rate $\bar r$ to be the link with the maximum latency under the optimal flow $\bar{f^*}$ for $\bar r$, i.e., 
$$k(\bar r) = \arg \max_{j:\bar{f_j^*}>0}\ell_j(\bar{f_j^*})\,.$$

\begin{lemma}
    \label{auxLem1}
    Given a predicted input rate $\bar r$ and the optimal flow $\bar{f^*}$ under $\bar r$, it holds that $\bar{f_i^*}>0$ for any $i \le k(\bar r)$, and $\bar{f_i^*}=0$ for any $i > k(\bar r)$. Moreover, if there exists a link $j$ with $\bar{f_j^*}>0$ and $a_j=0$, then $j=k(\bar r)$.
\end{lemma}

\begin{proof}
    We first show that there do not exist two links $i,j$, with $i<j$, such that $\bar f_i^*=0$ and $\bar f_j^*>0$. Suppose on the contrary that  $\bar f_i^*=0$ and $\bar f_j^*>0$. Since $\bar f^*$ is the optimal flow, the marginal cost of link $j$ should not exceed that of $i$, i.e., $\frac{d}{d f_j^*}\bar f_j^*\ell_j(\bar f_j^*) \leq \frac{d}{d f_i^*} \bar f_i^*\ell_i(\bar f_i^*)$, which means that  $2a_j\bar f_j^*+b_j \leq 2a_i\bar f_i^*+b_i = b_i$. This contradicts the fact that $b_i<b_j$ (recall that the links are sorted in increasing order of their constant factor). Therefore, we have established that there exists some index $k$ such that $\bar f_i^*>0$ for all $i\le k$ and  $\bar f_i^*=0$ for all $i>k$. 
    We further need to show that this $k$ is the $k(\bar r)$, i.e., that $\ell_k(\bar f_k^*)\geq \ell_i(\bar f_i^*)$, for all $i\le k$. Again we use the fact that $\bar f^*$ is the optimal flow, meaning that $2a_i\bar f_i^*+b_i = 2a_k\bar f_k^*+b_k$, for all $i\le k$. Then, for any $i\le k$, 
    $$\ell_i(\bar f_i^*)=a_i\bar f_i^*+b_i=\frac{2a_i\bar f_i^*+2b_i}{2}=\frac{2a_k\bar f_k^*+b_k+b_i}{2}=\ell_k(\bar f_k^*)-\frac{b_k-b_i}{2}\le \ell_k(\bar f_k^*)\,.$$

       To complete the proof we show if there exists a link $j$ with $\bar{f_j^*}>0$ and $a_j=0$, then for all $i>j$, $\bar{f_i^*}=0$. Suppose on the contrary that there exists $i>j$, such that $\bar{f_i^*}>0$. Since $i>j$, it holds that $b_i> b_j$, leading to 
       $2a_i\bar f_i^*+b_i > b_j= 2a_j\bar f_j^*+b_j$, which contradicts the fact that $\bar f^*$ is the optimal flow.
\end{proof}

\subsection{Constant Mechanisms} 
We define a class of parametrized CM that we call \textsc{Constant} mechanisms, that given a parameter $c$ and a predicted input rate $\bar{r}$, they increase the latency of the links to $c$ for flow greater than the optimal flow under $\bar{r}$, as long as $c$ is greater than the original latency. All those CM are consistent when $c\ge L$, where $L$ is as defined in Lemma~\ref{lemma:characterizationCons}.  
We then specify as \textsc{MinCharge} mechanism the \textsc{Constant} mechanisms with parameter $L$, which we show that is optimal among all consistent CM.

\begin{definition}[\textsc{Constant} Mechanisms]
Given any latency functions $(\ell_i)_{i \in M}$ and a predicted input rate $\bar{r}$, let $\bar{f}^*=(\bar{f_i^*})_{i\in M}$ be the optimal flow for $\bar{r}$. 
    For any parameter $c \geq \ell_i(\bar{f_i^*})$, for all $i\in M$, a \textsc{Constant} mechanism is defined for $i\geq k(\bar r)$ as $\hat{\ell}_i(x)=\ell_i(x)$, $\forall x$, and for $i<k(\bar r)$ as
    $$
\hat{\ell}_i(x)=
\begin{cases}
c+\frac{\varepsilon}{f^c_i-\bar{f}^*_i} (x-\bar{f}^*_i), \quad \bar{f}^*_i < x < f^c_i \\
\ell_i(x), \quad \text{otherwise}
\end{cases}
$$
where $\varepsilon>0$ is an arbitrarily small value\footnote{The addition of $\frac{\varepsilon}{f^c_i-\bar{f}^*_i} (x-\bar{f}^*_i)$ is only to ensure that the modified latency functions are strictly increasing. This technicality guarantees some nice properties: the CM admits a unique Nash flow and the link flows in the Nash flow are nondecreasing in the input rate \cite{CominettiDS24}.}, and $f^c_i$ denote the amount of flow at which the modified latency function equals the original latency, i.e., $\ell_i(f^c_i) = c+\varepsilon$.\footnote{Note that due to Lemma \ref{auxLem1}, each $\ell_i$ for $i<k(\bar r)$ is strictly increasing} 
\end{definition}

Focusing on a specific instantiation of the \textsc{Constant} mechanisms, we define the following mechanism as the \textsc{MinCharge} mechanism. 

\begin{definition}[\textsc{MinCharge} Mechanism]
\label{def:minChargeMech}
Given any latency functions $(\ell_i)_{i \in M}$ and a predicted input rate $\bar{r}$, let $\bar{f}^*=(\bar{f_i^*})_{i\in M}$ be the optimal flow for $\bar{r}$ and $L = \max_{i: \bar{f}_i^*>0} \ell_i(\bar{f}_i^*)$. 
    The \textsc{MinCharge} mechanism is defined as the \textsc{Constant} Mechanisms with parameter $L$, or more analytically for $i\geq k(\bar r)$ as $\hat{\ell}_i(x)=\ell_i(x)$, $\forall x$, and for $i<k(\bar r)$ as 
$$
\hat{\ell}_i(x)=
\begin{cases}
L+\frac{\varepsilon}{f^L_i-\bar{f}^*_i} (x-\bar{f}^*_i), \quad \bar{f}^*_i < x < f^L_i \\
\ell_i(x), \quad \text{otherwise}
\end{cases}
$$
where $\varepsilon>0$ is an arbitrarily small value, and $f^L_i$ is such that $\ell_i(f^L_i) = L+\varepsilon$.
\end{definition}

Next we provide a robustness guarantee for the \textsc{MinCharge} mechanism.

\begin{theorem}
\label{thm:Optimality}
    The \textsc{MinCharge} mechanism is consistent and $2$-robust on parallel-link networks with affine cost functions.
\end{theorem}

\begin{proof}
    The \textsc{MinCharge} mechanism trivially satisfies the requirements of Lemma \ref{lemma:characterizationCons}, hence it is consistent. We now analyze the behavior of the ePoA$(r,\bar{r})$ as a function of the input rate $r$ and the predicted input rate $\bar{r}$. We show that the worst-case inefficiency is realized for $r=\bar{r}+\epsilon$ for $\epsilon>0$ that goes to zero. We distinguish between the cases of $r\leq\bar{r}$ and $r>\bar{r}$. If $r\leq\bar{r}$, we show that ePoA$(r,\bar{r}) \leq$ PoA$(r)\leq 4/3$.

    \textbf{Case 1:} Overprediction ($r\leq\bar{r}$). Suppose that $f^*$ is the optimal flow under $r$ with respect to the original latency functions and $f$ denotes the Nash flow under $r$ with respect to the modified latency functions. In the overprediction case, for each link $i$ it holds that $f_i^*\leq \bar{f_i^*}$ and $f_i\leq \bar{f_i^*}$ \cite{algorithmica/ChristodoulouMP14, CominettiDS24}; thus, $\hat{\ell}_i(f_i^*)=\ell_i(f_i^*)$ and $\hat{\ell}_i(f_i)=\ell_i(f_i)$. Then, the variational inequality \cite{trsc.14.1.42/Dafermos,NesterovDePalma} holds for $f$ and $f^*$:

        $$\sum_{i} \ell_i(f_i) f_i \le \sum_{i} \ell_i(f_i)f_i^*\,.$$
    By following the proof of the celebrated result on the PoA of selfish routing due to \cite{RoughgardenT02}, we have that $\frac{\hat{C}_N(r)}{C_{opt}(r)} = 4/3$.

    \textbf{Case 2:} Underprediction ($r > \bar{r}$). We now consider the second case, where $r > \bar{r}$. Suppose again that $f^*$ and $f$ are the optimal flow w.r.t. the original latencies and Nash flow w.r.t. the CM, respectively, both under $r$. 

In the special case that $\ell_{k(\bar r)}$ is a constant, it should be $\ell_{k(\bar r)}=L$, and all the flow $\bar r -r$, will be routed through link $k(\bar r)$. Then, $$ePoA(r,\bar{r}) \leq \frac{rL}{C_{opt}(r)}\,.$$ 
    
     On the other hand, if $\ell_{k(\bar r)}$ is a strictly increasing function, then by Lemma~\ref{auxLem1}, every $\ell_{i}$ is strictly increasing for $i\leq k(\bar r)$.   In this case, every link incurs a latency higher than $L$ under $f$, where $L$ is the parameter of the \textsc{MinCharge} mechanism. To see this, note that for $r > \bar{r}$ there exists a link $i$ with $f_i>\bar{f}_i^*$ and therefore, $\hat{\ell}_i(f_i)>L$; then, if there exists a link $j$ with $\hat{\ell}_j(f_j)\leq L$, then there exists a sufficiently small $\epsilon>0$ such that  $\hat{\ell}_j(f_j+\epsilon)<\hat{\ell}_i(f_i)$, which violates the fact that $f$ is a Nash flow. 
    
    Having established that the minimum cost in any link under $r$ is more than $L$, we define $r_{max}$ to be the maximum rate for which the latency on each link in the Nash flow does not exceed $L+\varepsilon$ (as defined in Definition~\ref{def:minChargeMech}). 
    If $r>r_{max}$, all link flows in the Nash flow  are beyond the modification point, i.e., $f_i>f_i^L$ for all $i$ such that $\bar{f}_i^*>0$, and so $\hat{\ell}(f_i)=\ell(f_i)$, for all $i \in M$.
    It follows that $\hat{C}_N(r) = C_N(r)$. Hence, $ePoA(r) = PoA(r) \le 4/3$. Therefore, we focus on the interval $r \in [\bar{r}, r_{max}]$.

    In this regime, by omitting $\varepsilon$ as it is arbitrarily small, the ePoA is given by 
    $$ePoA(r,\bar{r}) = \frac{rL}{C_{opt}(r)}\,.$$ which we analyze below; note that this is the same bound as in the case of $\ell_{k(\bar r)}=L$. We first give the following claim to show that the ePoA is maximized when $r$ approaches $\bar{r}$.
    
    \begin{claim} \label{claim:derivative}
        $ \frac{r}{C_{opt}(r)}$ is a decreasing function of $r$. 
    \end{claim}

    \begin{proof}
        We take the derivative of the expression as a function of $r$, $\frac{C_{opt}(r) - r \frac{d}{dr}C_{opt}(r)}{C_{opt}(r)^2}$, and show that is non-positive. It is known that the optimal cost is $C_{opt}(r) = \frac{1}{\Lambda_j}(r^2+r\Gamma_j)-C_j$ \cite{algorithmica/ChristodoulouMP14}, where $\Lambda_j$, $\Gamma_j$ and $C_j$ are non-negative coefficients. Hence, it suffices to show that $r \frac{d}{dr}C_{opt}(r)\geq C_{opt}(r)$. From the expression for $C_{opt}(r)$, we compute:

        \begin{align*} r \frac{d}{dr}C_{opt}(r) = \frac{r}{\Lambda_j}(2r+\Gamma_j)=\frac{1}{\Lambda_j}(2r^2+\Gamma_jr) \geq C_{opt}(r).\end{align*}
\end{proof}

    Under the \textsc{MinCharge} mechanism, the ePoA$(r,\bar{r})$ is maximized for $r = \lim_{\epsilon\to 0^+} (\bar{r}+\epsilon)$. The ePoA is defined as:

        $$ePoA(r,\bar{r}) = \frac{\hat{C}_N(r)}{C_{opt}(r)}=\frac{rL}{C_{opt}(r)},$$

    where $L$ is the constant latency applied by the \textsc{MinCharge} CM. We analyze the derivative of ePoA with respect to $r$. By Claim \ref{claim:derivative}, $C_{opt}(r) < r \cdot C_{opt}(r)'$. Hence, the ePoA of the \textsc{MinCharge} CM is strictly decreasing for $r > \bar{r}$, and the maximum is attained in the limit as $r \rightarrow \bar{r}$.

Therefore for $r \in [\bar{r}, r_{max}]$ the ePoA$(r,\bar{r})$ is as follows:

\begin{equation}
\label{eq:ePoA_start}
    ePoA(r,\bar{r}) = \frac{\bar{r}L}{C_{OPT}(\bar{r})} = \frac{\bar{r}L}{\sum_{i:\bar f_i^*>0} \bar{f}_i^*\ell_i(\bar{f}_i^*)}
\end{equation}

We next show that for any $i\leq k(\bar r)$, which coincide with all links $i$ with $\bar{f}_i^*>0$, it holds $\ell_i(\bar{f}_i^*)\geq L/2$. Since $\bar{f}^*$ is an optimal flow, it holds that the marginal costs are equal in all links up tp $k(\bar r)$, therefore,  $2a_i\bar f_i^*+b_i=2a_k\bar f_k^*+b_k$. Then, 

$$\ell_i(\bar f_i^*) = a_i\bar f_i^*+b_i = a_k\bar f_k^*+\frac{b_k}2+\frac{b_i}2 = \ell_k(\bar f_k^*)-\frac{b_k}2+\frac{b_i}2\geq \ell_k(\bar f_k^*)-\frac{b_k}2 = L - \frac{b_k}2 \geq \frac L2\,,$$
where for the last equality we used the definition of $L$, and for the last inequality we used the fact that $L=\ell_k(\bar f_k^*)=a_k\bar f_k^*+b_k \geq b_k$.

We plug in that $\ell_i(\bar f_i^*)>\frac L2$ for all $i$ such that $\bar f_i^*>0$, to bound the optimal flow under $\bar{r}$:

$$C_{OPT}(\bar{r})=\sum_{i:\bar f_i^*>0} \bar{f}_i^*\ell_i(\bar{f}_i^*)\geq \frac L2 \sum_{i:\bar f_i^*>0} \bar{f}_i^* = \frac L2 \bar{r}\,.$$

Substituting this into equation~(\ref{eq:ePoA_start}) we get:

$$ePoA(r,\bar{r}) \leq \frac{\bar{r}L}{\frac L2 \bar{r}} = 2\,.$$ \end{proof}

We then show that the \textsc{MinCharge} CM is optimal with respect to robustness among all consistent CMs. For this we present an instance that precludes the existence of a consistent CM that is $(2 - \epsilon)$-robust for any $\epsilon > 0$.

\begin{lemma} [Lower bound]
\label{lemma:LBAffine}
    Any consistent CM is at least $2$-robust for affine cost functions. 
\end{lemma}

\begin{proof}

Let $\bar{r}=1$ be the predicted input rate. We consider the network with two parallel links with latencies $\ell_1(x)=x$ and $\ell_2(x)=2x+2-3\delta$, for sufficiently small $\delta>0$. We calculate the optimal flow $\bar{f}^*$ for $\bar{r}$, by setting the marginal costs equal, i.e., $2\bar{f}_1^* = 4\bar{f}_2^* + 2-3\delta$, where $\bar{f}_1^*+\bar{f}_2^*=1$. This gives that $\bar{f}_1^*=1-\frac{\delta}{2}$, and $\bar{f}_2^*=1-\bar{f}_1^*=\frac{\delta}{2}$.

Note that $\max_{i\in\{1,2\}} \ell_i(\bar{f}_i^*)=\ell_2(\bar{f}_2^*)=2-2\delta$. Consider any consistent CM; by Lemma \ref{lemma:characterizationCons}, this mechanism has to modify the latencies such that 

$$\hat{\ell}_i(\bar{f_i^*})=\ell_i(\bar{f_i^*})\mbox{,\; for } i\in\{1,2\} \mbox{\quad and  \quad } \liminf_{\epsilon \to 0^+} \hat{\ell}_1(\bar{f}_1^* + \epsilon) \ge 2-2\delta\,,$$

Let $f$ be a Nash flow of the CM. We consider the input rate $r$ so that it is the minimum amount of flow needed such that $\hat{\ell}_1(f_1)=\hat{\ell}_2(f_2)$. This happens right at the time that in the Nash flow link $1$ receives flow more than $\bar f_1^*$. Therefore, it holds that $r = \bar{f_1^*}+f_2\approx 1 + f_2$. In what follows, we set $\delta=0$ as it doesn't play any role. 
Let $f^*$ be the optimal flow under $r$ with respect to the original latency functions. We calculate $f^*$ by setting the marginal costs equal, i.e., $2f_1^* = 4f_2^* + 2$, where $f_1^*+f_2^*=r = 1+f_2$. This results in $f_1^* = 1+\frac{2}{3}f_2$ and $f_2^* = r- f_1^* = \frac 13 f_2$. Then, we bound the ePoA$(\bar{r})$ as follows:

$$ePoA\geq ePoA(r,\bar{r}) =\frac{\hat{C}_N(r)}{C_{opt}(r)}=\frac{r \hat{\ell}_2(f_2)}{f_1^*\ell(f_1^*)+f_2^*\ell(f_2^*)}\geq \frac{r \ell_2(f_2)}{f_1^{*2}+2f_2^{*2} +2f_2^*}\,.$$

By replacing $r=1+f_2$, $\ell(f_2)=2f_2+2$, $f_1^*=1+\frac 23 f_2$ and $f_2^*= \frac 13 f_2$, we get:
$$ePoA\geq \frac{2(1+f_2)^2}{(1+\frac 23 f_2)^2+2(\frac 13 f_2)^2 +\frac 23 f_2}\geq 2\,,$$
where the last inequality holds because the expression attains a minimum of $2$ at $f_2 = 0$.\end{proof}

The value for the modified latencies that minimizes the ePoA is precisely the one selected by the \textsc{MinCharge} mechanism. Combining Theorem \ref{thm:Optimality} and Lemma \ref{lemma:LBAffine}, we obtain the following corollary.

\begin{corollary}
    The \textsc{MinCharge} CM is optimal with respect to robustness among all consistent CMs.
\end{corollary}

\section{Smoothness and Error-Tolerant Design}\label{sec:smoothness}

In the learning-augmented framework, one seeks algorithms that perform optimally when supplied with accurate predictions, yet retain provable guarantees when predictions are imperfect. The notion of smoothness captures how the performance degrades as a function of the prediction error. In this section, we extend our CMs to be \emph{error-tolerant}, i.e., mechanisms whose approximation guarantee deteriorates gracefully as a function of the prediction error, up to a specified threshold. We derive a characterization of all consistent error-tolerant mechanisms (Lemma \ref{lemma:characterizationConsSmooth}). We then provide an approximation guarantee as a function of the prediction error, whenever the error is upper bounded by the error-tolerance threshold, and we show that this is tight for two links. Additionally, we establish a robustness bound in terms of $\bar\eta$ that the mechanism satisfies for arbitrary input rate. Combining the above, we derive our main result regarding smoothness in Theorem \ref{thm:errorTolerance}. 

\subsection{Characterization of Error-Tolerant Mechanisms}

We provide a full characterization of error-tolerant consistent CMs, extending the result of Lemma \ref{lemma:characterizationCons}. To achieve smoothness for an input rate in a restricted range around the predicted input rate $\bar{r}$, we require that the flow is routed exclusively through links without discontinuities. This goes in addition to the consistency of the CM, meaning that the two conditions of Lemma \ref{lemma:characterizationCons} are still necessary. Therefore, if the optimal flow for input rate $\bar{r}$ uses the first $k(\bar r)$ links\footnote{We refer the reader to  Lemma~\ref{auxLem1} for clarifications on which links are used in the optimal flow.}, the $k(\bar r)-1$ links come with discontinuities for more than the optimal flow, meaning that the extra input rate to $\bar{r}$ should be routed along the links from $k(\bar r)$ to $m$.  

\begin{lemma} [Characterization lemma] 
\label{lemma:characterizationConsSmooth}
    Given any latency functions $(\ell_i)_{i \in M}$, a predicted input rate $\bar{r}$, an error-tolerance threshold $\bar{\eta}$, and any CM, $(\hat{\ell}_i)_{i \in M}$, let $\bar{f}^* = (\bar{f}_i^*)_{i \in M}$ be the optimal flow for input rate $\bar{r}$. Suppose that $\bar{f}^{\bar\eta}$ is the Nash flow of the input rate $\bar{f}_{k(\bar{r})}^*+\bar\eta$ by considering the network with links $k(\bar r)$ through $m$ with the modified latency functions $(\hat{\ell}_i)_{i \in \{k(\bar r),\ldots m\}}$. We define as $\hat L_{\bar\eta}$ the maximum latency of any used link in $\bar{f}^{\bar\eta}$, and we further define $\bar{f}^{-\bar\eta}$ to be the Nash flow of the CM for input rate $\max\{\bar{r}-\bar\eta,0\}$. Then, the CM is consistent and error-tolerant, if and only if:

    \begin{itemize}
        \item for every $i\leq k(\bar r)$,  $\hat{\ell}_i(\bar{f}_i^*) = \ell_i(\bar{f}_i^*)$, and
        \item for every $i < k(\bar r)$, $\liminf_{\epsilon \to 0^+} \hat{\ell}_i(\bar{f}_i^* + \epsilon) \ge \hat L_{\bar\eta}$, and 
        \item for every $i\leq k(\bar r)$, $\hat{\ell}_i(x)$ is continuous for $x\in[\bar{f}^{-\bar\eta}_i, \bar{f}_i^*]$, and for every $i\geq k$, $\hat{\ell}_i(x)$ is continuous for $x\in [\bar{f}_i^*, \bar{f}^{\bar\eta}_i]$.  
    \end{itemize}    
\end{lemma}

\begin{proof}
    For simplicity let $k=k(\bar r)$. We first show that any CM satisfying the conditions of the statement is consistent. Since links $k+1$ to $m$ have higher $b$-values than $\ell_{k}(\bar{f}_{k}^*)$, otherwise they would be assigned a positive flow, it follows that $\hat L_{\bar{\eta}} \ge L$ ($L$ as defined in Lemma \ref{lemma:characterizationCons}). Therefore, any CM meeting the conditions of the lemma also satisfies the conditions of Lemma \ref{lemma:characterizationCons} and is consequently consistent.  
    Next, we prove that any CM satisfying the conditions of the lemma is error-tolerant. We show that for any $r \in [\max\{\bar{r}-\bar\eta,0\}, \bar{r} +\bar{\eta}]$, for any Nash flow $f$, no discontinuity appears, i.e., $f_i\leq \bar f_i^*$ for $i <k $. Suppose on the contrary that there exists some $i<k$ with $f_i > \bar f_i^*$, meaning that $\hat\ell_i(f_i)\geq \hat L_{\bar\eta}$. Then, there exists some link $j<k$ with $f_j<\bar f_j^*$ or some link $j\geq k$ with $f_j< \bar f^{\bar\eta}_j$. In any case, $\hat\ell_j(f_j)<\hat L_{\bar\eta}$, meaning that $f$ is not a Nash flow.

    We now proceed to show that any CM that is consistent and error-tolerant satisfies the conditions of the lemma. The first condition is necessary for consistency similarly to the proof of Lemma~\ref{lemma:characterizationCons}. 
    Regarding the second condition, suppose for the sake of contradiction that there exists a link $i<k$ such that $\liminf_{\epsilon\rightarrow0^+}{\hat{\ell}_i(\bar{f}_i^*+\epsilon)} < \hat L_{\bar{\eta}}$. We show that for the input rate $r=\bar{r}+\bar{\eta}$, any Nash flow $f$ has $f_i>\bar f_i^*$, which would violate smoothness. Suppose for the sake of contradiction that $f$ is a Nash flow that satisfies smoothness, which means that the aggregated link flow for the first $k-1$ links is at most $\bar{r}-\bar{f_k^*}$. Therefore, the rest of the flow, which is at least $\bar{f_k^*}+\bar\eta$ is routed via links $k$ though $m$ under $f$. Based on the definition of $\hat L_{\bar\eta}$, there exists at least one link among them with positive flow and latency at least $\hat L_{\bar\eta}$. Since we assumed $\liminf_{\epsilon\rightarrow0^+}{\hat{\ell}_i(\bar{f}_i^*+\epsilon)} < \hat L_{\bar{\eta}}$ for some $i<k$, $f$ cannot be a Nash flow. Therefore, the second condition of the lemma's statement is also necessary. The third condition is trivially necessary for error-tolerance.
\end{proof}

\subsection{\textsc{ErrorTolerant} Mechanism} 

We extend the \textsc{MinCharge} mechanism to achieve error-tolerance, by introducing a new mechanism, which we refer to as the \textsc{ErrorTolerant} mechanism. This mechanism can be viewed as the direct analog of the \textsc{MinCharge}, adapted to account for a bounded prediction error. Given an error-tolerance parameter $\bar{\eta}$, the CM sets a high latency for flow greater than the optimal flow under $\bar{r}$ to steer any additional flow towards links $k(\bar r)$ through $m$, ensuring that no link incurs a steep increase in latencies. Given an error-tolerance parameter $\bar{\eta}$, we instantiate the constant mechanism by setting $c = L_{\bar{\eta}}$ which is the minimum latency on links $k(\bar r)$ through $m$ when the additional flow $\bar{\eta}$ is routed through them.

\begin{definition}[\textsc{ErrorTolerant} Mechanism]
    Given any latency functions $(\ell_i)_{i \in M}$, a predicted input rate $\bar{r}$, and an error-tolerance threshold $\bar{\eta}$, let $\bar{f}^* = (\bar{f}_i^*)_{i \in M}$ be the optimal flow for $\bar{r}$. Suppose that $\bar{f}^{\bar\eta}$ is the Nash flow of the input rate $\bar{f}_{k(\bar r)}^* +\bar\eta$ by considering the network with links $k(\bar r)$ through $m$, with $(\ell_i)_{i \in \{k(\bar r),\ldots m\}}$. We define as $L_{\bar\eta}$ the latency of any used link in $\bar{f}^{\bar\eta}$. The \textsc{ErrorTolerant} mechanism is defined for any $i\geq k(\bar r)$ as $\hat{\ell}_i(x)=\ell(x), \forall x$, and for any $i< k(\bar r)$ as:

$$
\hat{\ell}_i(x)=
\begin{cases}
L_{\bar{\eta}}+\frac{\varepsilon}{f^{L_{\bar{\eta}}}_i-\bar{f}^*_i} (x-\bar{f}^*_i), \quad \bar{f}^*_i < x < f^{L_{\bar{\eta}}}_i \\ 
\ell_i(x), \quad \text{otherwise}
\end{cases}
$$
where $\varepsilon>0$ is an arbitrarily small value, and $f^{L_{\bar{\eta}}}_i$ is such that $\ell_i(f^{L_{\bar{\eta}}}_i) = L_{\bar{\eta}}+\varepsilon$.
\end{definition}

We then proceed to bound the approximation guarantee of the \textsc{ErrorTolerant} CM as a function of the prediction error. In fact, we show that the analysis of the \textsc{ErrorTolerant} approximation guarantee is tight for any $\eta \leq \bar{\eta}$, i.e., for $ r \in [\bar{r}, \bar{r}+\bar{\eta}]$ in the case of two links. We generalize our results to networks with any number of links, deriving upper bounds on both the approximation and robustness guarantees.

\subsection{Two Links} 
We begin with the case of two links. We establish a tight bound on the approximation guarantee when $\eta \leq \bar{\eta}$, and a tight bound on the robustness guarantee when $\eta > \bar{\eta}$ that the mechanism satisfies at all times. We first establish that, when considering the Nash and the optimal flow over two links, it is without loss of generality to set the latency of the first link as $\ell_1(x)=x$, and that of the second link as $\ell_2(x)=ax+b$, for some constants $a,b \ge 0$. Since only these two flows are relevant in the analysis of the PoA, it is without loss of generality to focus on such latencies.

\begin{lemma}
\label{lemma:Wlog_l1=x}
Consider any network $G$ of two links with latencies $\ell_1(x)=a_1x+b_1$ and $\ell_2(x)=a_2x+b_2$, and the network $G'$ of two links with latencies $\ell'_1(x)=x$ and $\ell'_2(x)=ax+b$, with $a=\frac{a_2}{a_1}$ and $b=\frac{b_2-b_1}{a_1}$. Then $f$ is a Nash (and respectively optimal) flow of $G$ if and only if it is a Nash (and respectively optimal) flow of $G'$.
\end{lemma}

\begin{proof}
    If $f$ is a Nash flow of $G$, then $\ell_1(f_1)=\ell_2(f_2)$, i.e., $a_1f_1+b_1=a_2f_2+b_2$, and so $\ell'_1(f_1)=f_1=\frac{a_2}{a_1}f_2+\frac{b_2-b_1}{a_1}=af_2+b=\ell'_2(f_2)$, which means that $f$ is a Nash flow in $G'$. The other direction can be shown similarly, so $f$ is a Nash flow of $G$ if and only if it is a Nash flow in $G'$.

    If $f$ is an optimal flow of $G$, then the marginal costs of the two links should be equal, i.e., $2a_1f_1+b_1=2a_2f_2+b_2$, and so $2f_1=\frac{2a_2}{a_1}f_2+\frac{b_2-b_1}{a_1}=2af_2+b$ meaning that the marginal costs of the two link of $G'$ are equal under $f$ and therefore $f$ is an optimal flow in $G'$. The other direction can be shown similarly, so $f$ is an optimal flow of $G$ if and only if it is an optimal flow of $G'$.
\end{proof}

\begin{lemma}
\label{lemma:PoAoriginal}
    For any given network $G$ with two links $\ell_1(x)=x$ and $\ell_2(x)=ax+b$, the PoA$(r)$ is given by:
    \begin{itemize}
        \item for $r<b/2, PoA(r) = 1 \,,$
        \item for $b/2\le r<b$, $$PoA(r) = \frac{(a+1)r^2}{(a+1)r^2-(r-\frac{b}{2})^2}\,,$$
        \item for $r \ge b$, $$PoA(r) = \frac{(ar+b)r}{(a+1)r^2-(r-\frac{b}{2})^2}\,.$$
    \end{itemize}
\end{lemma}

\begin{proof}
    For $r \in [0, b/2)$, the Nash and optimal flows use only the first link \cite{algorithmica/ChristodoulouMP14}, thus $PoA(r)=1$. 
    
    For $r \in [b/2,b)$, the optimal flow uses both links whereas Nash flow uses only the first link \cite{algorithmica/ChristodoulouMP14}. Hence, the cost at Nash flow is $C_N(r)=r^2$. To compute the cost of the optimal flow we use the fact that that the marginal costs are equal in the two links, i.e. $2f_1^* = 2af_2^*+b$. Together with $f_1^*+f_2^*=r$, we get that $f_1^*=\frac{2ar+b}{2(a+1)}$ and $f_2^*=\frac{2r-b}{2(a+1)}.$ Hence:

\begin{align*}
    C_{opt}(r,\bar r) & = {f_1^*}^2+a{f_2^*}^2+bf_2^* \\
    & = \frac{a^2r^2+\frac{b^2}{4}+arb+ar^2+a\frac{b^2}{4}-arb}{(a+1)^2}+\frac{br-\frac{b^2}{2}}{a+1} \\
    & = \frac{ar^2(a+1)+\frac{b^2}{4}(a+1)}{(a+1)^2}+\frac{br-\frac{b^2}{2}}{a+1} \\
    & = \frac{ar^2+br-\frac{b^2}{4}}{a+1} \\
    & = \frac{(a+1)r^2-\left(r^2-br+({\frac{b}{2}})^2\right)}{a+1} \\
    & = r^2-\frac{{\left(r-\frac{b}{2}\right)}^2}{a+1}\,.
\end{align*}
    
    For $r \ge b$, both the Nash and optimal flows use both links. In the Nash flow, latencies in the two links are equal and so $f_1=af_2+b$. Using $f_1+f_2=r$ we derive that $f_2=\frac{r-b}{a+1}$. The entire flow experiences the same cost, thus the total cost incurred in the Nash equilibrium is $C_N(r)=r(af_2+b)=r(a\frac{r-b}{a+1}+b) = r\frac{ar+b}{a+1}.$ Using the expression for the optimal cost computed above, the PoA follows.
    
\end{proof}
\begin{corollary}\label{cor:PoA43}
    For any given network $G$ with two links $\ell_1(x)=x$ and $\ell_2(x)=ax+b$, the PoA is given by:
    $$PoA=1+\frac 1{4a+3}\leq \frac 4 3\,.$$
\end{corollary}
\begin{proof}
    It is easy to verify that the PoA$(r)$ as given in Lemma~\ref{lemma:PoAoriginal} is increasing for $r\leq b$, attains its maximum value at $r=b$, and is decreasing for $r \ge b$. By setting $r=b$, we get the maximum PoA over all $r$ as  $PoA=1+\frac 1{4a+3}$. This is a decreasing function of $a$, therefore it gets its maximum for $a=0$, giving $PoA\leq 4/3$, as it is well known from the literature \cite{RoughgardenT02}.
\end{proof}
\begin{lemma}
\label{lemma:ErrorTolerantApxEta2}
    Consider any network $G$ of two links with $\ell_1(x)=x$, and $\ell_2(x)=ax+b$, for some constants $a,b \ge 0$. Given a predicted input rate $\bar r$ and an optimal flow $\bar f^*$ under $\bar r$, let $r'$ be the largest input rate such that the Nash flow on link 1 is less than $\bar f_1^*$, and $r''=f_1^{L_{\bar\eta}}+f_2^{L_{\bar\eta}}$ (where $f_i^{L_{\bar\eta}}$ is as defined in the \textsc{ErrorTolerant} mechanism). Given an error-tolerance threshold $\bar\eta$, for any input rate $r$, let the actual error $\eta=r-\bar r$ (which we allow to be negative), with $|\eta|\le\bar\eta$. Then, the \textsc{ErrorTolerant} mechanism achieves an approximation guarantee: 

    \begin{itemize}
        \item for $r\leq r'$ and $r \geq r''$, $ePoA(r,\bar r)=PoA(r)$ as in Lemma~\ref{lemma:PoAoriginal},
        \item for $r \in [r', \bar{r}+\bar\eta]$, 
        $$ePoA(r,\bar r) = 1 + a - \frac{ab(a+1)(\frac{b}{4}+\eta)}{a\eta^2+(a+1)b\eta+\frac{(a+1)b^2}{4}} \,,$$
        \item for $r \in [\bar{r}+\bar\eta, r'']$, $$ePoA(r,\bar r)= \frac{r\left(a\left(\frac{{\left(\bar r-\frac{b}{2}\right)}}{a+1}+\bar\eta\right)+b\right)}{r^2-\frac{{\left(r-\frac{b}{2}\right)}^2}{a+1}} \,.$$ 
    \end{itemize}
\end{lemma}

\begin{proof}
    {\bf Case (a)}: Note that for $r\leq r'$ and $r \geq r''$ the modifications of the cost functions do not affect the ePoA, therefore $ePoA(r,\bar r)=PoA(r)$.

    {\bf Case (b)}: It follows from the definition of our mechanism that there is enough capacity in both links such that any additional flow will not cause the social cost to increase abruptly. Define the difference between the optimal flow $f^*$ under the actual rate $r$ and the optimal flow $\bar f^*$ under the predicted rate $\bar{r}$ on link 1 as $\delta = f_1^* - \bar{f_1^*}$. It holds that $2f_1^*=2a(r-f_1^*)+b=2a(\bar r+\eta-f_1^*)+b$ and $2\bar f_1^*=2a(\bar r-\bar f_1^*)+b$, therefore, $\delta = \frac{a\eta}{a+1}.$ 

    Since $r \in [r', \bar r +\bar\eta]$, in the Nash flow $f_1=\bar f_1^*$, therefore,
    
    \begin{align*}
        ePoA(r, \bar r) & = \frac{\hat{C}_N(r)}{C_{opt}(r)} \\
        & =\frac{\bar{f}_1^*\hat{\ell}_1(\bar{f}_1^*)+(r-\bar{f}_1^*)\hat{\ell}_2(r-\bar{f}_1^*)}{C_{opt}(r)} \\
        & = \frac{(f_1^*-\delta)\ell_1(f_1^*-\delta)+(r-f_1^*+\delta)\ell_2(r-f_1^*+\delta)}{C_{opt}(r)} \\
        & = \frac{f_1^*\ell_1(f_1^*-\delta)-\delta\ell_1(f_1^*-\delta)+f_2^*\ell_2(r-f_1^*+\delta)+\delta\ell_2(r-f_1^*+\delta)}{C_{opt}(r)}.
    \end{align*}
    
Using the fact that the latencies satisfy $\ell_1(\bar{f}_1^*)=\ell_1(f_1^*-\delta)=\ell_1(f_1^*)-\delta$ and $\ell_2(r-\bar{f}_1^*)=\ell_2(r-f_1^*+\delta)=\ell_2(f_2^*)+a\delta$, we substitute these into the previous expression to obtain:

    \begin{align*}
        ePoA(r, \bar r) & = \frac{f_1^*(\ell_1(f_1^*)-\delta)-\delta(\ell_1(f_1^*)-\delta)+f_2^*(\ell_2(f_2^*)+a\delta)+\delta(l_2(f_2^*)+a\delta)}{C_{opt}(r)} \\
        & =  \frac{f_1^*\ell_1(f_1^*)+f_2^*\ell_2(f_2^*)+\delta(\ell_2(f_2^*)-\ell_1(f_1^*))+\delta(af_2^*-f_1^*)+\delta^2(a+1)}{C_{opt}(r)}\\
        & =  \frac{C_{opt}(r)+\delta(\ell_2(f_2^*)-\ell_1(f_1^*))+\delta(af_2^*-f_1^*)+\delta^2(a+1)}{C_{opt}(r)}.
    \end{align*}
    
Since the optimal flows $f^*$ satisfy $\ell_2(f_2^*)-\ell_1(f_1^*)=b/2$, we get:
    \begin{align*}
        ePoA(r) & = 1 + \frac{\delta b/2 + \delta b/2 - \delta b +\delta^2(a+1)}{C_{opt}(r)} \\
        & = 1 + \frac{\delta^2(a+1)}{C_{opt}(r)} \\
        & = 1 + \frac{a^2\eta^2}{(a+1)C_{opt}(r)}. 
    \end{align*}
    
We have already established that $f_1^*-\bar{f}_1^*=\frac{a\eta}{a+1}$. Moreover, we have noted that that the optimal and Nash flows begin to diverge once $\bar f_1^* \ge b/2$ (meaning that if $\bar f_1^* < b/2$ the mechanism would not make any changes). Combining these observations, we have that $f_1^* = \frac{a\eta}{a+1}+\bar{f}_1^*\ge\frac{a\eta}{a+1}+b/2$. Similarly, for the second link, we have $f_2^*-\bar{f}_2^*=\frac{\eta}{a+1}$, and hence, $f_2^* \ge \frac{\eta}{a+1}$. We proceed to derive a lower bound on the optimal cost $C_{opt}(r)$ using the aforementioned inequalities.

\begin{align*}
    C_{opt}(r) & = {f_1^*}^2+a{f_2^*}^2+bf_2^* \\
    & \ge \frac{a^2\eta^2}{(a+1)^2}+\frac{a\eta b}{a+1}+\frac{b^2}{4}+\frac{a\eta^2}{(a+1)^2}+\frac{b\eta}{a+1} \\
    & = \frac{a\eta^2+b\eta+a\eta b}{a+1}+\frac{b^2}{4}.
\end{align*}    

Substituting back, we have:

\begin{align*}
    ePoA(r) & \le 1+\frac{a^2\eta^2}{(a+1)\left(\frac{a\eta^2+b\eta+a\eta b}{a+1}+\frac{b^2}{4}\right)}\\
    & = 1+\frac{a^2\eta^2}{a\eta^2+b\eta+a\eta b+(a+1)\frac{b^2}{4}} \\
    & = 1+\frac{a^2\eta^2+a\eta b+a^2\eta b-a\eta b-a^2\eta b+a(a+1)\frac{b^2}{4}-a(a+1)\frac{b^2}{4}}{a\eta^2+b\eta+a\eta b+(a+1)\frac{b^2}{4}} \\
    & = 1 + a - \frac{ab(a+1)(\frac{b}{4}+\eta)}{a\eta^2+(a+1)b\eta+\frac{(a+1)b^2}{4}}
\end{align*}
For $\eta=0$, we get that $ePoA(r)=1$, while for $\eta \rightarrow \infty$, $ePoA(r)=1+a$.

This is tight because there exists an input rate $\bar{r}$ for which the inequalities are equalities. For $\bar{r}=b/2+\epsilon$ for $\epsilon$ arbitrarily small, $\bar{f_2^*}=0$, and therefore, the analysis is tight.

{\bf Case (c): } 
In this case the cost of the Nash flow is approximately $\hat{C}_N(r,\bar r)=r\hat{L}_{\bar{\eta}}=r(a(\bar f_2^*+\bar\eta)+b)$ based on the definition of the \textsc{ErrorTolerant} mechanism. Since $\bar f^*$ is the optimal flow for $\bar r$, it holds that $2(\bar r - \bar f_2^*)=2a\bar f_2^*+b$ and so 
$$\hat{C}_N(r,\bar r)=r\left(a\left(\frac{{\left(\bar r-\frac{b}{2}\right)}}{a+1}+\bar\eta\right)+b\right)\,.$$

Similarly to case (b) in Lemma \ref{lemma:PoAoriginal}, the cost of the optimum flow is $C_{opt}(r) = r^2-\frac{{\left(r-\frac{b}{2}\right)}^2}{a+1}$\,. Then, we derive,

\begin{align*}
    ePoA(r,\bar r) & = \frac{\hat{C}_N(r)}{C_{opt}(r)} = \frac{r\left(a\left(\frac{{\left(\bar r-\frac{b}{2}\right)}}{a+1}+\bar\eta\right)+b\right)}{r^2-\frac{{\left(r-\frac{b}{2}\right)}^2}{a+1}} 
\end{align*}
\end{proof}

\begin{lemma}
\label{lemma:rob2}
    Consider any network $G$ of two links with $\ell_1(x)=x$, and $\ell_2(x)=ax+b$, for some constants $a,b \ge 0$. Given any predicted input rate $\bar r$ and any error-tolerance threshold $\bar{\eta}$, the \textsc{ErrorTolerant} mechanism achieves the following robustness guarantee:

  $$ePoA = \max\{2, 1+a\}\,.$$

\end{lemma}

\begin{proof} 

{\bf Upper bound.} 
Consider the third case of Lemma~\ref{lemma:ErrorTolerantApxEta2}. The $ePoA(r,\bar r)$ is given by: 

$$
    ePoA(r,\bar r) = \frac{\hat{C}_N(r)}{C_{opt}(r)} = \frac{r\hat{L}_{\bar{\eta}}}{C_{opt}(r)} 
$$

Since $\hat{L}_{\bar{\eta}}=\left(a\left(\frac{{\left(\bar r-\frac{b}{2}\right)}}{a+1}+\bar\eta\right)+b\right)$ does not depend on $r$, by Claim $\ref{claim:derivative}$,  ePoA$(r,\bar r)$ is a decreasing function of $r$ and since $r\geq \bar r +\bar\eta$, we get: 

\begin{align*}
    ePoA(r,\bar r) & \le \frac{(\bar r +\bar\eta)\left(a\left(\frac{{\left(\bar r-\frac{b}{2}\right)}}{a+1}+\bar\eta\right)+b\right)}{(\bar r +\bar\eta)^2-\frac{{\left(\bar r +\bar\eta-\frac{b}{2}\right)}^2}{a+1}} \, .
\end{align*}

It is known that the optimal flow begins using link 2 at rate $b/2$ \cite{algorithmica/ChristodoulouMP14, isaac/WangYC16:ImprovedTax}. The derivative of the above expression is non-positive, implying that the expression is maximized when $\bar{r}=b/2$. Therefore, we set $\bar{r}$ with $b/2$ and get:

\begin{align*}
    ePoA(r,\bar r) & \le \frac{(b+a\bar{\eta})(\frac{b}{2}+\bar{\eta})}{\left(\frac{b(a+1)+2a\bar{\eta}}{2(a+1)}\right)^2+a\left(\frac{\bar{\eta}}{2(a+1)}\right)^2+\frac{b\bar{\eta}}{2(a+1)}} \\
    & = \frac{(b+a\bar{\eta})(\frac{b}{2}+\bar{\eta})}{\left(\frac{b}{2}+\frac{a\bar{\eta}}{a+1}\right)^2+a\left(\frac{\bar{\eta}}{a+1}\right)^2+\frac{b\bar{\eta}}{2(a+1)}} \\
    & = \left(1 + \frac{a}{b}\bar{\eta}\right)\cdot\frac{2b+4\bar{\eta}}{b+4\bar{\eta}+\frac{4a}{b(a+1)}\bar{\eta}^2}
\end{align*}

To analyze the bound more effectively, we divide both the numerator and denominator of the right-hand side fraction by $b$, expressing the bound as a function of $u = \frac{\bar \eta}{b} \geq 0:$ 

\begin{align*}
    ePoA(r,\bar r) & \le \frac{4au^2 + (4+2a)u + 2}{\frac{4a}{a+1}u^2+4u+1}\,.
\end{align*}

We inspect this function for $u \geq 0$. The expression evaluates to $2$ at $u=0$ and approaches $1+a$ as $u \to \infty$. When $a\geq 2$, the function is monotonically increasing. When $a < 2$, the function is decreasing on the interval $[0, \frac{1-a+\sqrt{3-a}}{2a}]$ and increasing thereafter. Therefore, the worst-case $ePoA(r,\bar r)$ is $\max\{2,1+a\},$ which is strictly worse than the $\frac{4}{3}$-robustness of Corollary \ref{cor:PoA43}. It is also strictly worse than the guarantee of case (b), since for $r \in [r',\bar r + \bar\eta]$, $ePoA(r,\bar r)$ is less than $1+a$, as the fraction in the expression of Lemma \ref{lemma:ErrorTolerantApxEta2} is positive.

Hence, case (c) gives the worst case robustness guarantee.

    {\bf Lower bound.} The analysis of the \textsc{ErrorTolerant} robustness guarantee is tight in the case of two links. It suffices to set $\bar{r}=b/2$; then, we have $\bar{f_1^*}=\frac{2a\bar{r}+b}{2(a+1)}=\bar{r}$, and $\bar{f_2^*}=0$. The \textsc{ErrorTolerant} CM dictates that:

        $$
        \hat{\ell}_1(x)=
        \begin{cases}
        a\bar{\eta}+b, \quad b/2 \le x < a\bar{\eta}+b \\
        \ell_1(x), \quad \text{otherwise}
        \end{cases}
        $$

    For $r=\bar{r}+\bar{\eta}=b/2+\bar{\eta}$, we have $f_1^*=\frac{2ar+b}{2(a+1)}=\frac{2a(b/2+\bar{\eta})+b}{2(a+1)}=\frac{ab+2a\bar{\eta}+b}{2(a+1)}=\frac{b}{2}+\frac{a\bar{\eta}}{a+1}$ and correspondingly, $f_2^*=r-f_1^*=\frac{\bar{\eta}}{a+1}$. The $ePoA$ is given by:

    \begin{align*}
        ePoA(r) & = \frac{\hat{C}_N(r)}{C_{opt}(r)} \\
        & = \frac{r\hat{L}_{\bar{\eta}}}{{f_1^*}^2+a{f_2^*}^2+bf_2^*}\\
        & = \frac{(b+a\bar{\eta})(\frac{b}{2}+\bar{\eta})}{\left(\frac{b(a+1)+2a\bar{\eta}}{2(a+1)}\right)^2+a\left(\frac{\bar{\eta}}{2(a+1)}\right)^2+\frac{b\bar{\eta}}{2(a+1)}} \\
    & = \frac{(b+a\bar{\eta})(\frac{b}{2}+\bar{\eta})}{\left(\frac{b}{2}+\frac{a\bar{\eta}}{a+1}\right)^2+a\left(\frac{\bar{\eta}}{a+1}\right)^2+\frac{b\bar{\eta}}{2(a+1)}} \\
    & = \left(1 + \frac{a}{b}\bar{\eta}\right)\cdot\frac{2b+4\bar{\eta}}{b+4\bar{\eta}+\frac{4a}{b(a+1)}\bar{\eta}^2}
    \end{align*}
\end{proof}

\subsection{Many links} 
We extend our results to networks with an arbitrary number of links, establishing upper bounds on the approximation and robustness guarantees, respectively. For simplicity of the presentation we set $k=k(\bar r)$.

\begin{theorem}
\label{thm:errorTolerance}
    We define $\Lambda_j$ and $\Gamma_j$ as $\Lambda_j = \sum_{i\le j}\frac{1}{a_i}$ and $\Gamma_j = \sum_{i \le j}\frac{b_i}{a_i}$, respectively. Let $\bar f^*$ be the optimal flow under some predicted input rate $\bar r$. Then, for any error-tolerance threshold $\bar\eta>0$, the \textsc{ErrorTolerant} coordination mechanism achieves an approximation of $ePoA(\eta) \le 1 + \Lambda_j(a_{max}+a_j) \left(1 - \frac{\Gamma_j\eta - C_j\Lambda_j}{\eta^2+\Gamma_j\eta-C_j\Lambda_j}\right)$ if $\eta \le \bar\eta$, and $2+2\bar\eta\frac{a_k}{b_k}$ otherwise, where $\eta$ is the prediction error.
\end{theorem}

We split the proof of Theorem~\ref{thm:errorTolerance} into Lemmas \ref{lemma:ErrorTolerantApxEtaMany}, \ref{lemma:ErrorTolerantApxEtaManyOverpredicted} and \ref{lemma:ErrorTolerantApxEtaBar}.

\begin{lemma}
\label{lemma:ErrorTolerantApxEtaMany}
    Consider any network $G$ of $m$ links. Given any input rate $\bar r$ and the optimal flow $\bar f^*$ for $\bar r$. Define $\Lambda_k$ and $\Gamma_k$ as $\Lambda_k = \sum_{i\le k}\frac{1}{a_i}$ and $\Gamma_k = \sum_{i \le k}\frac{b_i}{a_i}$, respectively. Given an error-tolerance threshold $\bar\eta$, and an actual error $\eta\le\bar\eta$, the \textsc{ErrorTolerant} mechanism achieves an approximation guarantee given by: 
    $$ePoA(\eta) \le 1 + \Lambda_k(a_{max}+a_k) \left(1 - \frac{\Gamma_k\eta - C_k\Lambda_k}{\eta^2+\Gamma_k\eta-C_k\Lambda_k}\right),$$
    where $C_k = \left(\sum_{i=1}^h\sum_{h=i}^k(b_h-b_i)^2\frac{1}{a_ha_i}\right)/4\Lambda_k$, and $r \in [\bar{r}, \bar{r}+\bar\eta]$. 
\end{lemma}
\begin{proof} 
Recall that by Lemma \ref{auxLem1}, all links with positive flow under $\bar f^*$ are exactly the ones with index at most $k$.

For each link $i$, we define the difference between the optimal flow $f^*$ under the actual rate $r$ and the optimal flow $\bar f^*$ under the predicted rate $\bar{r}$ as $\delta_i = f_i^* - \bar{f_i^*}$. Then, it holds that $\sum_i \delta_i = \eta$. Moreover, let $a_{max} = \max_i a_i$ denote the maximum $a_i$ coefficient of the latencies. 

Then, the cost in the Nash flow for rate $r$ is upper bounded by the cost where for each link $i<k$ the flow is $\bar f^*_i$ (since in any Nash flow with rate up to  $\bar r+\bar \eta$ no such link is used with more flow) and the rest of the rate is routed via link $k$ (in the Nash flow more links with indices greater than $k$ could be used, which would improve the total cost). Therefore, the ePoA is bounded by:

    \begin{align*}
        & ePoA(r) =\frac{\hat{C}_N(r)}{C_{opt}(r)} \\
        & \le \frac{\sum_{1\le i<k}\bar{f}_i^*\ell_i(\bar{f}_i^*)+(r - \sum_{1 \le i < k}\bar{f}_i^*)\ell_k(r-\sum_{1\le i <k}\bar{f}_i^*)}{C_{opt}(r)} \\
        & = \frac{\sum_{1\le i<k}f_i^*(\ell_i(f_i^*)-a_i\delta_i)-\sum_{1\le i<k}\delta_i(\ell_i(f_i^*)-a_i\delta_i)+(f_k^*+\sum_{1\le i<k}\delta_i)\ell_k(f_k^*+\sum_{1\le i<k}\delta_i)}{C_{opt}(r)} \\
        & = \frac{\sum_{1\le i<k}f_i^*(\ell_i(f_i^*)-a_i\delta_i)-\sum_{1\le i<k}\delta_i(\ell_i(f_i^*)-a_i\delta_i)}{C_{opt}(r)} \\
        & + \frac{f_k^*(\ell_k(f_k^*)+a_k\sum_{1\le i<k}\delta_i)+\sum_{1\le i<k}\delta_i\ell_k(f_k^*)+a_k\sum_{1\le i<k}\delta_i\sum_{1\le i<k}\delta_i}{C_{opt}(r)} 
    \end{align*}
    Observe that $\sum_{1 \le i <k}f_i^*\ell_i(f_i^*) + f_k^*\ell_k(f_k^*) \leq C_{opt}(r)$ and that for each two links $i,j$, it holds that their marginal costs are equal in $f^*$, so $2a_if^*_i+b_i=2a_jf^*_j+b_j$. Hence, we obtain:
    \begin{align}
        ePoA(r) & \le \frac{C_{opt(r)}+\sum_{1\le i<k}\delta_i(a_kf_k^*-a_if_i^*)+\sum_{1\le i<k}\delta_i(\ell_k(f_k^*)-\ell_i(f_i^*))}{C_{opt}(r)}\notag\\
        & +\frac{\sum_{1\le i<k}a_i\delta_i^2+a_k\sum_{1\le i<k}\delta_i\sum_{1\le i<k}\delta_i}{C_{opt}(r)} \notag\\
        & \le 1 + \frac{\sum_{1\le i<k}\delta_i(b_i-b_k)/2+\sum_{1\le i<k}\delta_i(b_k-b_i)/2+\sum_{1\le i<k}a_i\delta_i^2+\eta^2 a_k}{C_{opt}(r)} \notag\\
    \end{align}

    Simplifying, we have:
    \begin{align}
        & \le 1 + \frac{\sum_{1 \le i < k}a_i\delta_i^2+a_k\eta^2}{C_{opt}(r)} \label{eq:smoothIntermediateEq} \\
        & \le 1 + \frac{\sum_{1 \le i < k}\delta_i\sum_{1 \le i <k}a_i\delta_i+a_k\eta^2}{C_{opt}(r)} \notag \\
        & \le 1 + \frac{\eta^2(a_{max}+a_k)}{C_{opt}(r)} \notag \,.
    \end{align}

    Using the expression on the optimal cost, we obtain:

    \begin{align*}
        ePoA(r) & \le 1 + \frac{\eta^2(a_{max}+a_k)}{\frac{1}{\Lambda_k}(r^2+\Gamma_kr)-C_k} \\
        & \le 1+ \frac{\eta^2(a_{max}+a_k)}{\frac{1}{\Lambda_k}(\eta^2+\Gamma_k\eta)-C_k} \\
        & = 1 + \Lambda_k(a_{max}+a_k)\frac{\eta^2}{\eta^2+\Gamma_k\eta-\Lambda_kC_k} \\
        & = 1 + \Lambda_k(a_{max}+a_k) \left(1 - \frac{\Gamma_k\eta - C_k\Lambda_k}{\eta^2+\Gamma_k\eta-C_k\Lambda_k}\right)
    \end{align*}

    For $\eta=0$, we obtain $ePoA=1$, whereas as $\eta \rightarrow \infty$, we have $ePoA \le 1 + \Lambda_k(a_{max}+a_k)$. 
\end{proof}

\begin{lemma}
\label{lemma:ErrorTolerantApxEtaManyOverpredicted}
    Consider any network $G$ of $m$ links and any input rate $\bar r$. Then, for $r\leq \bar r$, ePoA$\le 4/3$. 
\end{lemma}

\begin{proof}
    Let $\bar f^*$ and $f^*$ be the optimal flow under $\bar r$ and $r$, respectively. with respect to the original latency functions. Notice that since $r\le \bar r$, for each link $i$ it holds that $f_i^*\leq \bar{f_i^*}$, and thus, $\hat{\ell}_i(f_i^*)=\ell_i(f_i^*)$. The same holds for the Nash flow $f$ under $r$, i.e., for each link $i$ it holds that $f_i\leq \bar{f_i^*}$, and thus, $\hat{\ell}_i(f_i)=\ell_i(f_i)$. Then, the variational inequality \cite{trsc.14.1.42/Dafermos,NesterovDePalma} implies that $\sum_{i} f_i\ell_i(f_i) \le \sum_{i} f_i^*\ell_i(f_i)\,.$
    By following the proof of the celebrated result on the PoA of selfish routing due to \cite{RoughgardenT02}, we have that $\frac{\hat{C}_N(r)}{C_{opt}(r)} = 4/3$.
\end{proof}

Next, we show that the \textsc{ErrorTolerant} mechanism maintains an approximation that is bounded by a function of $\bar\eta$ at all times, even when the prediction error is arbitrarily large.

\begin{lemma}
\label{lemma:ErrorTolerantApxEtaBar}
    Consider any network $G$ of $m$ links. Given any input rate $\bar r$ and the optimal flow $\bar f^*$ for $\bar r$. Then, for any error-tolerance threshold $\bar\eta$, the \textsc{ErrorTolerant} mechanism achieves an approximation of at most $2+2\bar\eta\frac{a_k}{b_k}$, i.e., ePoA$\le 2+2\bar\eta\frac{a_k}{b_k}$.
\end{lemma}

\begin{proof}
We will show the bound for any value of $r$, which would give the ePoA bound. For $r\leq \bar r$, by Lemma~\ref{lemma:ErrorTolerantApxEtaManyOverpredicted} we get ePoA$\le 4/3<2+2\bar\eta\frac{a_k}{b_k}$.

In the following claim we show an upper bound on the ePoA$(r,\bar r)$ for $r\in [\bar r, \bar r+\bar\eta]$, different from the one in Lemma~\ref{lemma:ErrorTolerantApxEtaMany}, that helps to show that in this case the bound is at most  $2+2\bar\eta\frac{a_k}{b_k}$, as well.

    \begin{claim}
        For any predicted input $\bar{r}$ and any $r \in [\bar{r}, \bar{r}+\bar\eta]$, it is $ePoA(r,\bar r)\leq 2+2\bar\eta\frac{a_k}{b_k}$.
    \end{claim}

    \begin{proof}
        For each link $i$, we define the difference between the optimal flow $f^*$ under the actual rate $r$ and the optimal flow $\bar f^*$ under the predicted rate $\bar{r}$ as $\delta_i = f_i^* - \bar{f_i^*}$, as in Lemma~\ref{lemma:ErrorTolerantApxEtaMany}. 
        We will use \eqref{eq:smoothIntermediateEq} of Lemma~\ref{lemma:ErrorTolerantApxEtaMany}, where $\eta\leq \bar\eta$ and $\delta_i\leq f_i^*$:
        \begin{align*}
        ePoA(r, \bar r) & \le 1 + \frac{\sum_{1 \le i < k}a_i\delta_i^2+a_k\eta^2}{C_{opt}(r)}\\
        & \le 1 + \frac{\sum_{1 \le i <k}a_i(f_i^*)^2+a_k\eta^2}{C_{opt}(r)}\,.
        \end{align*}

        For any two links $i,j$, with $f_i^*>0$ and $f_j^*>0$, it holds that $\ell_j(f_j^*)\geq \frac{2a_jf_j^*+b_j}{2}=\frac{2a_if_i^*+b_i}{2}$, so $\ell_j(f_j^*)\geq a_if_i^*$ and $\ell_j(f_j^*)\geq\frac{b_i}{2}$. Let $\ell=\min_{j\leq k}\ell_j(f_j^*)$, then $a_if_i^*\leq \ell$, for all $i\leq k$, and we will use that $\ell_i(f_i^*)\geq\frac{b_k}{2}$, for all $i$ with $f_i^*>0$:

        \begin{align*}
        ePoA(r, \bar r) & \le 1 + \frac{\sum_{1 \le i <k} f_i^* \ell+a_k\eta^2}{\sum_{i}f_i^*\ell_i(f_i^*) }\\
        & \le 1 + \frac{\sum_{1 \le i <k}f_i^* \ell_i(f_i^*) +a_k\eta^2}{\sum_{i}f_i^*\ell_i(f_i^*)}\\
        & \le 2 + \frac{a_k\eta^2}{\sum_{i}f_i^*\ell_i(f_i^*)}\\
        & \le 2 + \frac{a_k\eta^2}{\frac{b_k}2 \sum_{i} f_i^*}\\
        & \le 2 + \frac{a_k\eta^2}{\frac{b_k}2 \eta}\\
        & \leq 2+2\bar\eta\frac{a_k}{b_k}\,.
        \end{align*}
        
    \end{proof}
    
    We now consider the case where $r >\bar r + \bar\eta$. According to the \textsc{ErrorTolerant} mechanism, at any Nash equilibrium $f$, all links will experience the jump in the latency cost, i.e. $f_i > \bar{f}_i^*$ for all $i$. Let $f^{L_{\bar{\eta}}}_i$ be as defined in the class of \textsc{ErrorTolerant} mechanisms, i.e., $f^{L_{\bar{\eta}}}_i$ is such that $\ell_i(f^{L_{\bar{\eta}}}_i) \approx L_{\bar{\eta}}$.
    Let $r' = \sum_{i}{f^{L_{\bar{\eta}}}_i}$ be the maximum input rate that in the Nash flow the latency of all links is approximately $L_{\bar{\eta}}$. 
    
    We split the proof into two cases according to whether $r$ is greater or less than $r'$. Suppose first that $r> r'$ and by the definition of $r'$, there exists a link $j$ for which $\ell_j(f_j) > L_{\bar{\eta}}$. But since all link costs are continuous after $\bar{f}_i^*$, it holds that $\ell_i(f_i) = \ell_j(f_j) > L_{\bar{\eta}}$ for any link $i$. Thus, the cost of the modified Nash flow coincides with the cost of the Nash flow in the initial network, i.e. $\hat{C}_N(r) = C_N(r)$, hence $ePoA \leq \frac{4}{3}$.

    Now, suppose that $r \in (\bar r + \bar \eta, r']$. According to the definition of $r'$, the total cost satisfies $\hat{C}_N(r) \leq rL_{\bar{\eta}}+\varepsilon$, where $\varepsilon>0$ is as defined in the \textsc{ErrorTolerant} mechanism; since $\varepsilon$ is arbitrarily small, we set it to zero for simplicity. Hence:

    \begin{align*}
        ePoA(r, \bar{r}) & = \frac{\hat{C}_N(r)}{C_{opt}(r)} = \frac{rL_{\bar{\eta}}}{C_{opt}(r)}.
    \end{align*}

    By Claim \ref{claim:derivative} the ratio increases as $r$ decreases, and since we are in the case that $r\geq \bar r$, it holds that 
    $$ePoA(r, \bar{r}) \le \frac{{\bar{r}}L_{\bar{\eta}}}{C_{opt}(\bar{r})}\,.$$

    Based on the definition of $L_{\bar{\eta}}$, it is upper bounded by the cost of link $k$ with flow $\bar f^*_k + \bar\eta$. Moreover, for any link $i$ with $\bar f_i^*>0$, $\ell_i(\bar f_i^*)\geq \frac{2a_i\bar f_i^*+b_i}{2}=\frac{2a_k\bar f_k^*+b_k}{2}\geq \frac{\ell_k(\bar f_k^*)}{2}$. By plugging in the above:
    
    \begin{align*}
        ePoA(r, \bar{r}) & \leq  \frac{\bar{r}(a_k(\bar f_k^*+\bar\eta)+b_k)}{\bar{r}\frac{\ell_k(\bar f_k^*)}{2}} \\
        & = 2\frac{\ell_k(\bar f_k^*) +a_k\bar\eta}{\ell_k(\bar f_k^*)}\\
        & \leq 2+2\bar\eta\frac{a_k}{b_k}\,.
    \end{align*}
\end{proof}

\section*{Acknowledgments}

This work has been partially supported by project MIS 5154714 of the National Recovery and Resilience Plan Greece 2.0 funded
by the European Union under the NextGenerationEU Program.

\bibliographystyle{alpha}  
\bibliography{references} 

\newcommand{\etalchar}[1]{$^{#1}$}
\begin{thebibliography}{GLMM10}

\bibitem[ABG{\etalchar{+}}24]{mor/AgrawalBGOT24}
Priyank Agrawal, Eric Balkanski, Vasilis Gkatzelis, Tingting Ou, and Xizhi Tan.
\newblock Learning-augmented mechanism design: Leveraging predictions for facility location.
\newblock {\em Math. Oper. Res.}, 49(4):2626--2651, 2024.

\bibitem[AFJ{\etalchar{+}}15]{AzarFJMS15}
Yossi Azar, Lisa Fleischer, Kamal Jain, Vahab~S. Mirrokni, and Zoya Svitkina.
\newblock Optimal coordination mechanisms for unrelated machine scheduling.
\newblock {\em Oper. Res.}, 63(3):489--500, 2015.

\bibitem[AO07]{mor/AcemogluOzdaglar07}
Daron Acemoglu and Asuman~E. Ozdaglar.
\newblock Competition and efficiency in congested markets.
\newblock {\em Math. Oper. Res.}, 32(1):1--31, 2007.

\bibitem[BGS24]{nips/BalkanskiGS24:randomizedFL}
Eric Balkanski, Vasilis Gkatzelis, and Golnoosh Shahkarami.
\newblock Randomized strategic facility location with predictions.
\newblock In {\em NeurIPS}, 2024.

\bibitem[BGT23]{ITCS/BalkanskiGT23}
Eric Balkanski, Vasilis Gkatzelis, and Xizhi Tan.
\newblock Strategyproof scheduling with predictions.
\newblock In {\em {ITCS}}, volume 251 of {\em LIPIcs}, pages 11:1--11:22. Schloss Dagstuhl - Leibniz-Zentrum f{\"{u}}r Informatik, 2023.

\bibitem[BGT24]{nips/Barak0T24:MAC}
Zohar Barak, Anupam Gupta, and Inbal Talgam{-}Cohen.
\newblock {MAC} advice for facility location mechanism design.
\newblock In {\em NeurIPS}, 2024.

\bibitem[BGTZ24]{sigecom/BalkanskiGTZ24:OnlineMD}
Eric Balkanski, Vasilis Gkatzelis, Xizhi Tan, and Cherlin Zhu.
\newblock Online mechanism design with predictions.
\newblock In {\em {EC}}, page 1184. {ACM}, 2024.

\bibitem[BIKM14]{BhattacharyaIKM14}
Sayan Bhattacharya, Sungjin Im, Janardhan Kulkarni, and Kamesh Munagala.
\newblock Coordination mechanisms from (almost) all scheduling policies.
\newblock In {\em {ITCS}}, pages 121--134. {ACM}, 2014.

\bibitem[BMW56]{Beckmann1956-}
Martin Beckmann, Bartlett McGuire, and Christopher~B. Winsten.
\newblock {\em Studies in the economics of transportation}.
\newblock 1956.

\bibitem[BS94]{transci/BernsteinS94:UserEq}
David Bernstein and Tony~E. Smith.
\newblock Equilibria for networks with lower semicontinuous costs: With an application to congestion pricing.
\newblock {\em Transp. Sci.}, 28(3):221--235, 1994.

\bibitem[CCG{\etalchar{+}}15]{ColeCGMO15}
Richard Cole, Jos{\'{e}}~R. Correa, Vasilis Gkatzelis, Vahab~S. Mirrokni, and Neil Olver.
\newblock Decentralized utilitarian mechanisms for scheduling games.
\newblock {\em Games Econ. Behav.}, 92:306--326, 2015.

\bibitem[CDR03]{stoc/ColeDR03}
Richard Cole, Yevgeniy Dodis, and Tim Roughgarden.
\newblock Pricing network edges for heterogeneous selfish users.
\newblock In {\em {STOC}}, pages 521--530. {ACM}, 2003.

\bibitem[CDS24]{CominettiDS24}
Roberto Cominetti, Valerio Dose, and Marco Scarsini.
\newblock The price of anarchy in routing games as a function of the demand.
\newblock {\em Math. Program.}, 203(1):531--558, 2024.

\bibitem[CGS24]{ChristodoulouGkatzelisSgouritsa}
George Christodoulou, Vasilis Gkatzelis, and Alkmini Sgouritsa.
\newblock Resource-aware cost-sharing methods for scheduling games.
\newblock {\em Oper. Res.}, 72(1):167--184, 2024.

\bibitem[CK24]{ijcai/CaragiannisK24}
Ioannis Caragiannis and Georgios Kalantzis.
\newblock Randomized learning-augmented auctions with revenue guarantees.
\newblock In {\em {IJCAI}}, pages 2687--2694. ijcai.org, 2024.

\bibitem[CKN09]{tcs/ChristodoulouKN09:CMs}
George Christodoulou, Elias Koutsoupias, and Akash Nanavati.
\newblock Coordination mechanisms.
\newblock {\em Theor. Comput. Sci.}, 410(36):3327--3336, 2009.

\bibitem[CKS18]{icalp/Colini-Baldeschi18}
Riccardo Colini{-}Baldeschi, Max Klimm, and Marco Scarsini.
\newblock Demand-independent optimal tolls.
\newblock In {\em {ICALP}}, volume 107 of {\em LIPIcs}, pages 151:1--151:14. Schloss Dagstuhl - Leibniz-Zentrum f{\"{u}}r Informatik, 2018.

\bibitem[CKST24]{sigecom/Colini-Baldeschi24}
Riccardo Colini{-}Baldeschi, Sophie Klumper, Guido Sch{\"{a}}fer, and Artem Tsikiridis.
\newblock To trust or not to trust: Assignment mechanisms with predictions in the private graph model.
\newblock In {\em {EC}}, pages 1134--1154. {ACM}, 2024.

\bibitem[CMP14]{algorithmica/ChristodoulouMP14}
Giorgos Christodoulou, Kurt Mehlhorn, and Evangelia Pyrga.
\newblock Improving the price of anarchy for selfish routing via coordination mechanisms.
\newblock {\em Algorithmica}, 69(3):619--640, 2014.

\bibitem[CSV24]{neurips:CSV24}
George Christodoulou, Alkmini Sgouritsa, and Ioannis Vlachos.
\newblock Mechanism design augmented with output advice.
\newblock In {\em NeurIPS}, 2024.

\bibitem[Daf80]{trsc.14.1.42/Dafermos}
Stella Dafermos.
\newblock Traffic equilibrium and variational inequalities.
\newblock {\em Transportation Science}, 14(1):42–54, 1980.

\bibitem[DPN{\etalchar{+}}98]{NesterovDePalma}
Andre De~Palma, Yurii Nesterov, et~al.
\newblock Optimization formulations and static equilibrium in congested transportation networks.
\newblock 1998.

\bibitem[FJM04]{focs/FleischerJM04}
Lisa Fleischer, Kamal Jain, and Mohammad Mahdian.
\newblock Tolls for heterogeneous selfish users in multicommodity networks and generalized congestion games.
\newblock In {\em {FOCS}}, pages 277--285. {IEEE} Computer Society, 2004.

\bibitem[GKST22]{sigecom/GkatzelisKOlliasSgouritsaTan22}
Vasilis Gkatzelis, Kostas Kollias, Alkmini Sgouritsa, and Xizhi Tan.
\newblock Improved price of anarchy via predictions.
\newblock In {\em {EC}}, pages 529--557. {ACM}, 2022.

\bibitem[GLMM10]{GairingLuckingMavronicolasMonien10}
Martin Gairing, Thomas L{\"{u}}cking, Marios Mavronicolas, and Burkhard Monien.
\newblock Computing nash equilibria for scheduling on restricted parallel links.
\newblock {\em Theory Comput. Syst.}, 47(2):405--432, 2010.

\bibitem[GPS21]{GkatzelisPountourakisSgouritsa21}
Vasilis Gkatzelis, Emmanouil Pountourakis, and Alkmini Sgouritsa.
\newblock Resource-aware cost-sharing mechanisms with priors.
\newblock In {\em {EC}}, pages 541--559. {ACM}, 2021.

\bibitem[GST25]{soda/Gkatzelis0T25}
Vasilis Gkatzelis, Daniel Schoepflin, and Xizhi Tan.
\newblock Clock auctions augmented with unreliable advice.
\newblock In {\em {SODA}}, pages 2629--2655. {SIAM}, 2025.

\bibitem[HM83]{HaurieMarcotte1983}
Alain Haurie and Patrice Marcotte.
\newblock On the relationship between nash - cournot and wardrop equilibria.
\newblock {\em Networks}, 15:295--308, 1983.

\bibitem[HSV19]{eor/HarksSV19}
Tobias Harks, Marc Schr{\"{o}}der, and Dries Vermeulen.
\newblock Toll caps in privatized road networks.
\newblock {\em Eur. J. Oper. Res.}, 276(3):947--956, 2019.

\bibitem[ILMS09]{tcs/ImmorlicaLMS09}
Nicole Immorlica, Li~(Erran) Li, Vahab~S. Mirrokni, and Andreas~S. Schulz.
\newblock Coordination mechanisms for selfish scheduling.
\newblock {\em Theor. Comput. Sci.}, 410(17):1589--1598, 2009.

\bibitem[KK04a]{focs/KarakostasK04}
George Karakostas and Stavros~G. Kolliopoulos.
\newblock Edge pricing of multicommodity networks for heterogeneous selfish users.
\newblock In {\em {FOCS}}, pages 268--276. {IEEE} Computer Society, 2004.

\bibitem[KK04b]{KarakostasKolliopoulos04}
George Karakostas and Stavros~G. Kolliopoulos.
\newblock The efficiency of optimal taxes.
\newblock In {\em {CAAN}}, volume 3405 of {\em Lecture Notes in Computer Science}, pages 3--12. Springer, 2004.

\bibitem[Kol13]{Kollias13}
Konstantinos Kollias.
\newblock Nonpreemptive coordination mechanisms for identical machines.
\newblock {\em Theory Comput. Syst.}, 53(3):424--440, 2013.

\bibitem[KP99]{KoutsoupiasP99}
Elias Koutsoupias and Christos~H. Papadimitriou.
\newblock Worst-case equilibria.
\newblock In {\em {STACS}}, volume 1563 of {\em Lecture Notes in Computer Science}, pages 404--413. Springer, 1999.

\bibitem[LV21]{jacm/LykourisV21}
Thodoris Lykouris and Sergei Vassilvitskii.
\newblock Competitive caching with machine learned advice.
\newblock {\em J. {ACM}}, 68(4):24:1--24:25, 2021.

\bibitem[LWZ24]{sigecom/LuW024:auctions}
Pinyan Lu, Zongqi Wan, and Jialin Zhang.
\newblock Competitive auctions with imperfect predictions.
\newblock In {\em {EC}}, pages 1155--1183. {ACM}, 2024.

\bibitem[MNS07]{sigecom/MahdianNS07}
Mohammad Mahdian, Hamid Nazerzadeh, and Amin Saberi.
\newblock Allocating online advertisement space with unreliable estimates.
\newblock In {\em {EC}}, pages 288--294. {ACM}, 2007.

\bibitem[MP07]{MarcottePatriksson}
Patrice Marcotte and Michael Patriksson.
\newblock {\em Chapter 10 Traffic Equilibrium}, volume~14, page 623–713.
\newblock 12 2007.

\bibitem[MV22]{cacm/MitzenmacherV22}
Michael Mitzenmacher and Sergei Vassilvitskii.
\newblock Algorithms with predictions.
\newblock {\em Commun. {ACM}}, 65(7):33--35, 2022.

\bibitem[Pat15]{patriksson2015traffic}
Michael Patriksson.
\newblock {\em The traffic assignment problem: models and methods}.
\newblock Courier Dover Publications, 2015.

\bibitem[Rou03]{Roughgarden03Indep}
Tim Roughgarden.
\newblock The price of anarchy is independent of the network topology.
\newblock {\em J. Comput. Syst. Sci.}, 67(2):341--364, 2003.

\bibitem[RT02]{RoughgardenT02}
Tim Roughgarden and {\'{E}}va Tardos.
\newblock How bad is selfish routing?
\newblock {\em J. {ACM}}, 49(2):236--259, 2002.

\bibitem[vFH13]{DBLP:journals/mor/FalkenhausenH13}
Philipp von Falkenhausen and Tobias Harks.
\newblock Optimal cost sharing for resource selection games.
\newblock {\em Math. Oper. Res.}, 38(1):184--208, 2013.

\bibitem[WYC16]{isaac/WangYC16:ImprovedTax}
Te{-}Li Wang, Chih{-}Kuan Yeh, and Ho{-}Lin Chen.
\newblock An improved tax scheme for selfish routing.
\newblock In {\em {ISAAC}}, volume~64 of {\em LIPIcs}, pages 61:1--61:12. Schloss Dagstuhl - Leibniz-Zentrum f{\"{u}}r Informatik, 2016.

\bibitem[XL22]{ijcai/XuLu22}
Chenyang Xu and Pinyan Lu.
\newblock Mechanism design with predictions.
\newblock In {\em {IJCAI}}, pages 571--577. ijcai.org, 2022.

\end{thebibliography}

\newpage
\appendix
\section{Missing Statements and Proofs from Sections \ref{sec:prelims} and \ref{sec:CMs}}

The following lemma is a statement from \cite{algorithmica/ChristodoulouMP14}. We present the proof here for completeness.

\begin{lemma}
    \label{auxLem0}
    Consider a parallel-link network $G=(V,E)$, and suppose that two links $i$ and $j$ share the same constant term, i.e., $b_i = b_j$. Construct a new parallel-link network $G'=(V,E')$, with $E' = E\setminus\{i,j\}\cup\{e\}$, where $e$ is a link that \say{merges} $i$ and $j$ with latency function $\ell_e(x) = \frac{a_ia_j}{a_i+a_j}x+b_i$. Then, if a flow $f$ is a Nash flow (respectively, an optimal flow) in $G$ then $f'$, with $f'_s=f_s$, for all $s\notin \{i,j\}$, and $f'_e=f_i+f_j$, is a Nash (optimal) flow in $G'$. Moreover, $C(f)=C(f')$.
\end{lemma}

\begin{proof}
    Suppose first that $f$ is a Nash flow of $G$, then the latency on all used links must be equal (Wardrop's first principle). That is, $\ell_i(f_i)=\ell_j(f_j)=\ell_s(f_s)$ for all $s$. Let $\ell_s(f_s)=\ell$. Then, $f_i=\frac{\ell-b_i}{a_i}$ and $f_j=\frac{\ell-b_j}{a_j}=\frac{\ell-b_i}{a_j}$.  Then, 
    
    $$\ell_e(f_e)=\frac{a_ia_j}{a_i+a_j}(f_i+f_j)+b_i=\frac{a_ia_j}{a_i+a_j}\left(\frac{1}{a_i}+\frac{1}{a_j}\right)(\ell-b_i)+b_i=\ell\,,$$
    which is equal to $\ell_s(f'_s)$ for all $s$, and therefore $f'$ is a Nash flow in $G'$. Moreover, $C(f)=\ell\sum_s f_s=\ell \sum_s f'_s=C(f')$.

    Similarly, if $f$ is an optimal flow of $G$, it should hold that the marginal costs are equal at each link, i.e., $\frac{d}{df_s} f_s \ell_s(f_s)$ is the same for all $s$; let this value be $\ell$. Then, $2a_if_i+b_i=2a_jf_j+b_j=\ell$, and so, $f_i=\frac{\ell-b_i}{2a_i}$ and $f_j=\frac{\ell-b_j}{2a_j}=\frac{\ell-b_i}{2a_j}$. The marginal cost at link $e$ is 
    $$2\frac{a_ia_j}{a_i+a_j}(f_i+f_j)+b_i=\frac{a_ia_j}{a_i+a_j}\left(\frac{1}{a_i}+\frac{1}{a_j}\right)(\ell-b_i)+b_i=\ell\,,$$

    which is equal to the marginal cost of any link $s$ in $G'$, and therefore $f'$ is an optimal flow in $G'$. Moreover, 
    \begin{align*}
    f_e\ell_e(f_e)&=(f_i+f_j)\left(\frac{a_ia_j}{a_i+a_j}(f_i+f_j)+b_i\right)\\
    &=f_i\left(\frac{a_ia_j}{a_i+a_j}(f_i+f_j)+b_i\right)+f_j\left(\frac{a_ia_j}{a_i+a_j}(f_i+f_j)+b_i\right)\\
    &= f_i\left(\frac{a_ia_j}{a_i+a_j}\left(f_i+\frac{2a_if_i}{2a_j}\right)+b_i\right)+f_j\left(\frac{a_ia_j}{a_i+a_j}\left(\frac{2a_jf_j}{2a_i}+f_j\right)+b_i\right)\\
    &= f_i\left(a_if_i+b_i\right)+f_j\left(a_jf_j+b_i\right)\\
    &=f_i\ell_i(f_i)+f_j\ell_j(f_j)\,,
    \end{align*}
    
    and therefore, $C(f)=C(f')$.
\end{proof}

\end{document}